\documentclass[11pt]{article}
\usepackage{algorithm}
\usepackage{algorithmic}
\usepackage{fullpage}
\usepackage{epsfig}
\usepackage{xspace}
\usepackage{theorem}
\usepackage{graphicx}
\usepackage{graphics}
\usepackage{colordvi}
\usepackage{multirow}
\usepackage{amsfonts, amstext, amssymb, amsmath, mathrsfs, wasysym}

\newenvironment{proofof}[1]{\noindent{\bf Proof of #1.}}%
        {\hspace*{\fill}$\Box$\par\vspace{4mm}}

\newcommand{\be}{\begin{enumerate}}
\newcommand{\ee}{\end{enumerate}}
\newcommand{\bd}{\begin{description}}
\newcommand{\ed}{\end{description}}
\newcommand{\bi}{\begin{itemize}}
\newcommand{\ei}{\end{itemize}}

\newtheorem{lemma}{Lemma}[section]

\newtheorem{theorem}{Theorem}

\newtheorem{claim}{Claim}

\newenvironment{proof}{\smallskip \noindent {\bf Proof:}}{\hfill\stopproof}
\def\stopproof{\square}
\def\square{\vbox{\hrule height.2pt\hbox{\vrule width.2pt height5pt \kern5pt
\vrule width.2pt} \hrule height.2pt}}

\newcommand{\ra}{\rightarrow}
\renewcommand{\phi}{\varphi}
\newcommand{\eps}{\epsilon}

\newcommand{\R}{\ensuremath{\mathbb R}}

\newcommand{\E}[1]{\text{\bf E}[#1]}
\newcommand{\Var}[1]{\text{\bf Var}[#1]}
\newcommand{\Exp}[2]{\text{\bf E}_{#1} [#2]}
\renewcommand{\Pr}[1]{\text{\bf Pr}\left [#1\right]}

\newcommand{\U}{\mathcal{U}}
\renewcommand{\v}{\bf{v}}

\newcommand{\M}{\mathcal{M}}

\renewcommand{\R}{\bf{R}}

\newcommand{\I}{\mathcal{I}}
\renewcommand{\P}{\mbox{\sf P}}

\newcommand{\rev}{{\bf Rev}}

\newcommand{\junk}[1]{}

\renewcommand{\u}{{\bf U}}
\newcommand{\ut}[1]{{\bf U}_{#1}}

\renewcommand{\E}[1]{{{\bf E}\left[#1\right]}}
\newcommand{\OPT}{{\rm {OPT}}}

\newcommand{\Ex}{{\bf E}}
\newcommand{\OPTs}{\mbox{\tiny OPT}}

\renewcommand{\P}{{\mathcal P}}
\newcommand{\sm}{\mbox{\tiny sm}}
\newcommand{\lrg}{\mbox{\tiny lg}}
\newcommand{\hg}{\mbox{\tiny hg}}
\newcommand{\poly}{\mbox{\tiny poly}}
\newcommand{\Rev}{{\mathcal R}}
\newcommand{\Prob}{{\mathbb Pr}}
\newcommand{\DSIC}{\mbox{\tiny DSIC}}
\newcommand{\BIC}{\mbox{\tiny BIC}}
\newcommand{\soft}{\mbox{\tiny soft}}
\newcommand{\expt}{\mbox{\tiny exp}}
\newcommand{\rounds}{\mbox{\tiny rounds}}
\newcommand{\copies}{\mbox{\tiny copies}}

\newcommand{\hard}{\mbox{\tiny hard}}
\newcommand{\D}{\mathcal{D}}

\newcommand{\ifshort}[1]{}

\date{}

\begin{document}

\title{Mechanism Design and Risk Aversion}

\author{Anand Bhalgat \ \ \ \ \ \ \ Tanmoy Chakraborty \ \ \ \ \ \ \ Sanjeev Khanna\\
Univ. of Pennsylvania \ \ \ \ \  Harvard University\ \ \ \ \ \ \ \ Univ. of Pennsylvania\\
~~bhalgat@seas.upenn.edu \ \ tanmoy@seas.harvard.edu\ \    sanjeev@cis.upenn.edu}

\begin{titlepage}
\thispagestyle{empty}
\maketitle
\thispagestyle{empty}

\begin{abstract}
We develop efficient algorithms to construct utility maximizing mechanisms in the presence of risk averse players (buyers and sellers) in Bayesian single parameter and multi-parameter settings. We model risk aversion by a concave utility function, and players play strategically to maximize their expected utility. Bayesian mechanism design has usually focused on maximizing expected revenue in a {\em risk neutral} environment, {\em i.e.} where all the buyers and the seller have linear utility, and  no succinct characterization of expected utility maximizing mechanisms is known even for single-parameter multi-unit auctions.

We first consider the problem of designing optimal DSIC (dominant strategy incentive compatible) mechanism for a risk averse seller in
the case of multi-unit auctions, and we give a poly-time computable deterministic sequential posted pricing mechanism (SPM) that for any $\eps > 0$, yields a $(1-1/e-\eps)$-approximation to the expected utility of the seller in an optimal DSIC mechanism. Our result is based on a novel application of a correlation gap bound, along with {\em splitting} and {\em merging} of random variables to redistribute probability mass across buyers. This allows us to reduce our problem to that of checking feasibility of a small number of distinct configurations, each of which corresponds to a covering LP. A feasible solution to the LP gives us the distribution on prices for each buyer to use in a randomized SPM. We get a deterministic SPM by sampling from this randomized SPM. Our techniques extend to the multi-parameter setting with unit demand buyers.

We next consider the setting when buyers as well as the seller are risk averse, and the objective is to maximize the seller's expected utility.  We design a truthful-in-expectation mechanism whose utility is a $\left(\left(1-\frac{1}{e}\right)^2\times\max\left(1-\frac{1}{e}, 1-\frac{1}{\sqrt{2\pi k}}\right)\right)$-approximation to the optimal BIC mechanism under two mild assumptions: (a) ex post individual rationality and (b) no positive transfers. Our mechanism consists of multiple rounds. It considers each buyer in a round with small probability, and when a buyer is considered, it allocates an item to the buyer according to payment functions that are computed using stochastic techniques developed for DSIC mechanisms. Lastly, we consider the problem of revenue maximization for a risk neutral seller in presence of risk averse buyers, and give a poly-time algorithm to design an optimal mechanism for the seller.

We believe that the techniques developed in this work will be useful in handling other stochastic optimization problems with a concave objective function.

\end{abstract}

\end{titlepage}

\newpage

\section{Introduction}\label{intro}

Bayesian mechanism design has usually focused on maximizing expected revenue in a {\em risk neutral} environment, {\em i.e.} where all the buyers and the seller have linear utility, and choose their strategy with the aim of maximizing their {\em expected payoff}. However, since the payoff is a random outcome that depends on other players' valuations and strategies, there is risk associated with it. A standard model \cite{A65,P64} that captures risk aversion assumes that a player has a non-decreasing concave utility function $\u:(-\infty,\infty) \rightarrow (-\infty,\infty)$, so that when the payoff obtained is $R$, the player's utility is $\u(R)$. The player may choose to express various levels of risk aversion by specifying a suitable concave function as his utility, and then his aim becomes to maximize his expected utility. While mechanism design in a risk neutral (linear utility) environment is well understood for multi-unit auctions, many properties tend to break down in the presence of risk aversion (concave utility). In this paper, we develop efficient algorithms to compute mechanisms in the presence of risk averse players. We mainly focus on the prominent single parameter setting of multi-unit auctions.

\vspace*{-0.4cm}
\paragraph{Risk neutral Seller, Risk Averse Buyers:}
Let us first analyze the effect of risk aversion among buyers when the seller is risk neutral, i.e. he wants to maximize his expected revenue. This has been the predominant model for studying risk aversion in mechanism design. Myerson's characterization \cite{M81} of the optimal auction design does not apply when buyers are risk averse. In particular, revenue equivalence \cite{M81} does not hold, and an optimal dominant strategy incentive compatible (DSIC) mechanism may generate less expected revenue than an optimal Bayesian incentive compatible (BIC) mechanism. This is because when buyers are risk averse, the seller can extract greater expected revenue by offering a deterministic payment scheme to the buyers and charging extra for this {\em insurance}. As a specific example, consider buyers with {\em constant absolute risk aversion}, {\em i.e.} exponentially diminishing marginal utility. For any DSIC mechanism (such as Myerson's), one can construct a corresponding BIC mechanism with the same allocation curves: if a buyer $i$ reports valuation $v$, charge a deterministic payment $\alpha(i,v)$ (perhaps zero) if the buyer does not get the item, and $\beta(i,v)$ if he gets the item. For the latter mechanism to be truthful, it must satisfy {\em utility equivalence} \cite{M83,M87,H05}: the net expected utility of the buyer in the latter mechanism must match the former when his value is $v$, for each $v$. However, if his expected payment in the latter mechanism remains equal to the former, then his net expected utility will go up, because concave utility makes the deterministic payment more preferable. Thus the expected payment must be higher in the latter mechanism, implying greater revenue. Note that the direct revelation principle still holds, so there is a BIC mechanism that generates as much expected revenue as any Bayes-Nash equilibrium\footnote{Any equilibrium of the mechanism where buyers play strategically to maximize their utility.}.

Maskin and Riley~\cite{MR84} characterized optimal BIC mechanism in this setting for selling a single item when buyers' value distributions are IID. It is assumed that the buyers' utility functions are known to the mechanism designer (the seller), since truthfulness itself depends on these functions.  Another well-known result states that under some natural assumptions on the buyers' utility functions, first-price auction with reserve (specifically, Bayes-Nash equilibrium of this auction) generates greater revenue than second-price auction with the same reserve \cite{M83,MR84,M87}. These results show that the presence of risk averse buyers should have substantial effect on optimal mechanism design.

\begin{table}
\centering
\begin{minipage}{\textwidth}
\begin{tabular}{| p{3.5cm} || p{5cm} || p{2.8cm} | p{2.9cm} |}
\hline
\multirow{2}{*}{
\hspace{-15pt}
\begin{tabular}{l}
    Type of risk \\
    environment
    \end{tabular}
    }
&
\multicolumn{1}{c||}{Comparison with Optimal DSIC}
&
\multicolumn{2}{c|}{Comparison with Optimal BIC}
\\
\cline{2-4}
&
\multicolumn{1}{c||}{Poly-time DSIC \footnotemark[1]}
&
\multicolumn{1}{c|}{Poly-time TIE \footnotemark[2]}
&
\multicolumn{1}{c|}{Poly-time BIC \footnotemark[2]}
\\
\hline
Risk neutral seller, risk neutral buyers
&
\centering 1 [Myerson '81]
&
\centering
1 \hspace{2cm}
[Myerson '81]
&
1
[Myerson '81]
\\
\hline
Risk neutral seller, risk averse buyers
&
\centering 1 [Myerson '81]
&
\centering  $\gamma(k)$ \footnotemark[3] [Theorem~\ref{thm:bic_risk_neutral}]
&
$1$
[Theorem~\ref{thm:bic_risk_neutral}]
\\
\hline
Risk averse seller, risk neutral buyers
&
\centering $(1-1/e - \eps)$ [Theorem~\ref{thm:multi-unit}]
&
\centering $(1-1/e)^2\gamma(k) - \eps$ 
[Theorem~\ref{thm:non-iid}]\footnotemark[4]
&
 1 [Eso-Futo '99]
\\
\hline
Risk averse seller, risk averse buyers
&

\centering $(1-1/e - \eps)$
 [Theorem~\ref{thm:multi-unit}]
&
\centering $(1-1/e)^2\gamma(k) - \eps$ 
[Theorem~\ref{thm:non-iid}]\footnotemark[4]
&
$(1-1/e)^2\gamma(k) - \eps$ 
[Theorem~\ref{thm:non-iid}]\footnotemark[4]
\\
\hline
\end{tabular}

\footnotetext[1]{{\footnotesize Need to know seller's utility function. Independent of buyers' utility functions as long as they are non-decreasing.}}
\footnotetext[2]{{\footnotesize Need to know both seller and buyers' utility functions.}}
\footnotetext[3]{{\footnotesize $\gamma(k) = (1 - \frac{k^k}{k! e^k})$. $\gamma(1)=1-1/e$, and it approaches $(1 - \frac{1}{\sqrt{2\pi k}})$ for large $k$.}}
\footnotetext[4]{{\footnotesize Improves to $(1-1/e) \gamma(k) - \eps$ for IID buyers (Theorem~\ref{thm:iid}). Further, the factor improves to $(1-1/e-\eps)$ if $k\ge 1/\eps^3$. Here, comparison is made only against optimal BIC satisfying: (i) ex-post IR, and (ii) no positive transfers.}}
\end{minipage}
\label{table:results}
\caption{Summary of approximation results for $k$-unit auctions.}
\end{table}

Our results for a $k$-unit auction are summarized in Table \ref{table:results}. We first
design a poly-time algorithm to compute a  BIC mechanism for a risk neutral seller when buyers are risk averse with publicly known utility functions. This result extends the work of Maskin and Riley \cite{MR84}, in a computational sense, to the general setting of multi-unit auctions with non-identical distributions. Our algorithm is a linear program developed using a general form of Border's inequality. Further, we design a
a poly-time computable randomized truthful-in-expectation mechanism which is a $\gamma(k)$-approximation to the utility optimal BIC mechanism. Here, $\gamma(k) = (1 - \frac{k^k}{k! e^k})$; $\gamma(1)=1-1/e$, and approaches $(1 - \frac{1}{\sqrt{2\pi k}})$ for large $k$. A mechanism is said to be {\em truthful in expectation} (TIE) if truth-telling maximizes each buyer's expected utility, regardless of other buyers' bids (expected utility is measured only over the random bits used by the mechanism). This is a stronger truthfulness requirement than BIC, but weaker than DSIC \footnote{In this paper, we require that in a DSIC mechanism, truth-telling must maximize buyers' utility regardless of the mechanism's random bits. This distinction between DSIC and TIE will play a crucial role below.}, and does not rely on all buyers sharing the same belief about each other. TIE mechanisms can still generate greater expected revenue than DSIC mechanisms. In effect, this result {\em bounds the gap between TIE mechanisms and BIC mechanisms}.

\vspace*{-0.4cm}
\paragraph{Risk Averse Seller, Risk Averse Buyers:}
Next, we consider the scenario where the seller as well as the buyers are risk averse. Risk aversion from a seller's perspective has received relatively less attention, and no work has considered both sides to be risk averse. Eso and Futo \cite{EF99} designed an optimal BIC mechanism for a risk averse seller when {\em buyers are risk neutral}. In this setting, the seller can {\em transfer the entire risk} to risk neutral buyers, and obtain the expected revenue of Myerson's mechanism in every realization. A DSIC mechanism cannot do such a risk transfer, and as such, the gap between optimal BIC and optimal DSIC is unbounded, as illustrated below.
\smallskip
\noindent {\em Example.}\ Consider an instance with two buyers and two items. Each buyer has valuation $1$ for the item w.p. $\epsilon$ and $0$ otherwise. The seller's utility function $\u$ is as follows: $\u(t)=\min\{t,\epsilon\}$. The utility optimal DSIC mechanism sets a price of $1$ to each buyer, and gets utility $\epsilon$ with probability $2\epsilon(1-\epsilon)\le 2\eps$, otherwise its utility is $0$. So the expected utility of an optimal DSIC mechanism is at most $2\eps^2$. If the first buyer is risk neutral, then we can design a BIC mechanism as follows: charge the first buyer $\eps$ in every realization (even when his value is zero), and set a price of $1$ to the second buyer. If the second buyer pays up $1$ (which happens w.p. $\eps$), then pay $1$ dollar to the first buyer. The first buyer never gets the item, and the mechanism is incentive compatible for the first buyer. The seller gets a revenue of $\epsilon$ in every realization, so his expected utility is $\eps$. Therefore the gap is unbounded as $\eps\rightarrow 0$.
\smallskip

However, the extent to which the seller can transfer its risk to the buyers depends upon buyers' utility functions. So the result of ~\cite{EF99} does not hold when buyers are risk averse. We design TIE mechanisms that are constant approximation to a utility-optimal BIC mechanism for a risk averse seller. We restrict our comparison only to BIC mechanisms that do not allow {\em positive transfers} ({\em i.e.} there is no payment from the seller to any buyer in any realization), and are {\em ex post individually rational}. Our approximation factor is $(1-1/e)^2\gamma(k)$ ($(1-1/e) \gamma(k)$ for IID buyers), and approaches $(1-1/e)$ (using a slightly different algorithm) as $k$ becomes large. This implies a constant upper bound on the gap between TIE mechanisms and BIC mechanisms without positive transfer.


\vspace*{-0.4cm}
\paragraph{DSIC mechanisms for Risk Averse Sellers:}
So far, we have designed BIC or TIE mechanisms assuming that the players' utility functions are known to the designer. Though this has been a standard assumption through most of the literature on risk aversion, it is usually the seller who designs the mechanism, and it is not practical to assume that the seller knows all the buyers' utility functions. Unfortunately, it is impossible in general to even check if a mechanism is BIC or even TIE (or identify a Bayes-Nash equilibrium of a given mechanism) without knowing the buyers' utility functions. In contrast, 
{\em DSIC mechanisms are independent of buyers' utility functions} as long as buyers' utility functions are non-decreasing. So we do not need to know the buyers' utility functions to compute a utility-optimal DSIC mechanism for the seller.

Clearly, for a risk-neutral seller, Myerson's mechanism remains an optimal DSIC mechanism even when buyers are risk averse. However, that is not true for a risk averse seller. Further, a virtual value maximization approach does not apply when the seller has at least two units of inventory\footnote{If there is only one item to sale, then at most one buyer pays in any realization, and the seller's utility can be maximized by scaling the bid values using the utility function.}. This is because the contributions of different buyers cannot be counted separately when the seller has non-linear utility. We design a poly-time computable $(1-1/e-\eps)$-approximate DSIC mechanism. Our mechanism is a sequential posted-pricing mechanism (SPM), that decides a price for each buyer, and then makes take-it-or-leave it offers to the buyers in decreasing order of prices, till inventory runs out. Posted pricing mechanisms have been studied extensively for maximizing expected revenue ({\em eg.} \cite{CHMS10,CEGMM10_spm,Y11}).

Sundarajan and Yan \cite{SY10} designed DSIC mechanisms for multi-unit auctions for a risk-averse seller when buyers' valuation functions are {\em regular}.  They focused only on mechanisms that do not depend even on the seller's own utility function. While this may be  an attractive property, the optimal mechanism is no longer well-defined in this case (even neglecting  computational limitations), and it forces the approximation guarantees to be weaker -- $1/8$-approximation for regular distributions ($1/2$ when there is unlimited supply of items), and there is a lower bound instance implying  unbounded gap for general distributions. Non-regular distributions are not uncommon -- any distribution with more than one mode is non-regular. Risk aversion is particularly important in the presence of such high variance distributions, and it is a reasonable approach for a seller to decide upon his own utility function and then use our algorithm to design a mechanism.

As a final note, we are able to extend our techniques to give a constant approximation to an optimal deterministic DSIC mechanism in a multi-parameter setting, namely, when there are multiple distinct items and unit-demand buyers (see Appendix \ref{sec:unit-demand}).

\subsection{Overview of Techniques}
We first design a DSIC mechanism for a risk averse seller, which turns out to be a relatively simpler problem than competing against optimal BIC mechanism. We establish our main probabilistic tools in the process, which later get used for our BIC results. For DSIC mechanisms, we first argue that a $(1-\frac{1}{e})$-approximate (randomized) SPM can be obtained by using the same price distribution as that offered to each buyer in the optimal mechanism, except that the prices are now set independently (see Lemma~\ref{lem:exists}). The argument uses the {\em correlation gap} bound of Agrawal et. al. \cite{ADSY11} for submodular objectives. 
However, this is only an existential result, since getting the SPM requires oracle access to a utility-optimal DSIC mechanism.

Our main technical contribution is to show that it suffices (with same loss factor of $(1-1/e)$) to match the optimal mechanism only in the {\em sum of sale probabilities over all buyers, and not the sale probability for each buyer}, at every price. That is, any two mechanisms that match in this {\em coarse footprint} will have approximately equal expected utility. This property follows from a generalization of the correlation gap bound in \cite{ADSY11}, which not only introduces independence but also redistributes probability mass across variables (see Lemma~\ref{lem:mass-general}). The redistribution is achieved by {\em splitting} and {\em merging} random variables to transform one given mechanism to another that matches the coarse footprint. Using a sophisticated classification of prices, we show that it suffices to match an even coarser footprint containing only constant number of parameters, which define a {\em configuration}. The algorithm finds a feasible solution for each configuration using a {\em covering LP}. Then, it simulates these SPMs, one for each feasible configuration, to choose one with the highest expected utility. 

To design a BIC mechanism when the seller as well as the buyers are risk averse, the techniques developed for DSIC mechanisms can be used to establish that if  allocation and payment functions of the optimal mechanism across buyers are made independent, and inventory constraints removed, the utility will be at least $(1-1/e)\OPT_{\BIC}$. However, to convert such a {\em soft} mechanism into a mechanism that strictly satisfies the inventory constraint is not easy: if we restrict the allocation to buyers with {\em top $k$} payments in a realization of a soft mechanism, a function which is submodular, the resulting mechanism is no longer BIC. Further,  distributions on the revenue from any two allocations in the mechanism are incomparable, so restricting to {\em first $k$} allocations in a realization of a soft mechanism can be arbitrarily bad. To overcome this problem, we develop a mechanism with $L\rightarrow\infty$ {\em rounds}, such that in each round, each buyer is ignored with a high probability of $(1-1/L)$. We show that the revenue from each allocation in this mechanism has identical distribution, and the loss in the expected utility caused by imposing the hard inventory constraint is bounded.


\vspace*{-0.4cm}
\paragraph{Organization:} Rest of the paper is organized as follows.
In Section \ref{prelim}, we provide some background material. In
Sections \ref{sec:multi-unit} and \ref{sec:bayes-nash}, we present our DSIC and BIC mechanisms for risk averse seller, respectively. Our BIC mechanisms for the risk neutral seller and risk averse buyers appear in Appendix \ref{sec:risk-neutral-seller}, and our result for multi-parameter unit-demand setting appears in Appendix \ref{sec:unit-demand}.

\section{Preliminaries}\label{prelim}
\noindent
{\bf Single Parameter Multi-Unit Auctions:} The seller provides a single type of item (or service), of which he has $k$ identical copies. There are $n$ buyers $\{1,2,\ldots n\}$, who have some private value for that service. Let buyer $i$ have a valuation of $v_i$ for the item (and he can consume only one unit), which is drawn, independent of other buyers' valuations, from a known distribution with cdf $F_i(x) = \Pr{v_i\leq x}$. We refer to $\v = (v_1,v_2\ldots v_n)$ as the {\em valuation vector}. 

\smallskip
\noindent
{\bf Revenue, Utility and Optimality:} The {\em revenue} $\rev(\M,\v)$ of a mechanism $\M$, when the realized valuation vector is $\v$, is the sum of payments from each buyers. The expected revenue of a mechanism $\rev(\M)$ is $\Exp{\v}{\rev(\M,\v)}$. In this work, we assume that the seller has a monotonically increasing  concave utility function $\u$, which also satisfies $\u(0)=0$.  The utility of the mechanism is $\u(\rev(\M,\v))$, and the expected utility of the mechanism is $\u(\M)=\Exp{\v}{\u(\rev(\M,\v))}$. Let  $\OPT_{\DSIC}$ and  $\OPT_{\BIC}$ denote the expected utility of a utility-optimal DSIC and BIC mechanisms respectively.  A mechanism is said to be an {\em $\alpha$-approximation} to optimal DSIC (or BIC) mechanism if $\u(\M) \geq \alpha \OPT_{\DSIC}$ (or $\u(\M) \geq \alpha \OPT_{\BIC}$).


\smallskip
\noindent
{\bf DSIC Mechanisms:}
It is well-known ({\em eg.} \cite{M81}) that a DSIC mechanism sets a (possibly randomized) price for buyer $i$ based on $v_{-i}$ but independent of $v_i$, and buyer $i$ gets an item if and only if his valuation exceeds this price. Given this characterization, it is easy to see that as long as a buyer has a non-decreasing utility function, he will report truthfully in a DSIC mechanism, for any realization of valuation vector and random bits of the mechanism. Note that random bits do not help a DSIC mechanism obtain greater utility, since the definition of DSIC implies that truthfulness must hold even if the random bits were revealed prior to submitting bids. So there is a utility-optimal DSIC mechanism which is deterministic.

\smallskip
\noindent
{\bf Buyer's Risk Aversion and BIC Mechanisms:} Each buyer $i$ is associated with a publicly known monotone concave utility function $\U_i$ (defined on the value of item received minus payment) with $\U_i(0)=0$. 
A BIC mechanism is associated with two functions $h(\cdot,\cdot,\cdot)$ and $g(\cdot,\cdot,\cdot)$:
$h(i,j,v)$ is the probability that for valuation $v$, buyer $i$ is allocated an item for a payment of $p_j$, and $g(i,j,v)$ is the probability that he pays $p_j$ and is not allocated an item for valuation $v$. We refer to these two functions as the {\em payment functions} of the mechanism. We note that the allocation and payment of a buyer  is possibly correlated with other buyers' payments, allocations as well as their valuations. Thus, a mechanism is BIC if and only if for each $i,v,v'$, we have
\[\begin{array}{l}
\sum_{j} \left(\U_i(v-p_j)h(i,j,v) + \U_i(-p_j)g(i,j,v)\right)\ge \sum_{j} \left(\U_i(v-p_j)h(i,j,v')+ \U_i(-p_j)g(i,j,v')\right)
\end{array}\]
We note, given any buyer $i$, we allow his payment to be randomized rather than a fixed value as a function of buyer $i$'s valuation and whether he gets an item. This strictly gives more power to a risk averse seller maximizing his expected utility. This is in contrast to the setting considered by Maskin and Riley~\cite{MR84}, where it suffices to assume that buyer $i$'s payment for valuation $v$ is a fixed value as a function of $v$ and whether he gets the item.

We define a {\em soft randomized sequential mechanism} as a mechanism without inventory limit that arranges buyers in an arbitrary order, asks each buyer for his valuation one-by-one. If the buyer $i$'s reported valuation is $v$, the mechanism allocates an item to him independently w.p. $\sum_j h(i,j,v)$. If he is allocated an item, then the seller charges him $p_j$ w.p. $\frac{h(i,j,v)}{\sum_l h(i,l,v)}$. When he is not allocated an item, he pays $p_j$ w.p. $\frac{g(i,j,v)}{\sum_l g(i,l,v)}$. {\em Randomized sequential mechanisms} are same as {\em soft randomized sequential mechanisms} with an exception that they stop after running out of inventory. We note that if a {\em soft randomized sequential mechanism} is BIC, then the corresponding {\em randomized sequential mechanism} is also BIC.

\smallskip
\noindent
{\bf Stochastic Dominance:} Given two non-negative distributions $\D_1$ and $\D_2$, we say $\D_1$ {\em stochastically dominates} $\D_2$, denoted by $\D_1\succeq\D_2$, if $\forall a\ge 0$, $\Prob_{X\AC \D_1}(X_1\ge a) \ge \Prob_{X\AC \D_2}(X_1\ge a)$.
We note an important property of concave functions in the following lemma.
\begin{lemma}\label{lem:dominance}
Given any non-decreasing concave function $\u$, and three independent non-negative random variables $X, Y_1, Y_2$, let $\D_1$ and $\D_2$ be the distributions of $Y_1$ and $Y_2$ respectively. If $\D_1\succeq\D_2$, then $\Ex_{X, Y_1\AC\D_1}\left[\u(X+Y)-\u(Y)\right]\le \Ex_{X, Y_1\AC\D_2}\left[\u(X+Y)-\u(Y)\right]$.
\end{lemma}

\section{Risk Averse Seller: DSIC Mechanisms for Multi-Unit Auctions}\label{sec:multi-unit}
In this section, we construct DSIC mechanisms for a risk averse seller. 
The following theorem summarizes our result.
\begin{theorem}\label{thm:multi-unit}
For multi-unit auctions, there is a poly-time computable deterministic SPM with expected utility at least $(1-\frac{1}{e}-\epsilon) \OPT$, for any $\eps>0$, where $\OPT$ is the expected utility of an optimal DSIC mechanism.
\end{theorem}
We first prove the existence of an SPM that achieves a $(1-1/e)$-approximation to the optimal expected utility (Section~\ref{sec:existence}), however this result is not constructible and does not lead to an efficient implementation. We then identify a set of sufficient properties of $(1-1/e-\eps)$-approximate mechanisms that enables us to construct a poly-time algorithm (Section~\ref{sec:algorithm}).

\subsection{Existence of a $(1-1/e)$-approximate SPM}\label{sec:existence}
Given a set $S=\{x_1,x_2\ldots x_n\}$ of non-negative real number, let $\max_i \{ x_1,x_2\ldots x_n\}$ denote the $i^{th}$ largest value in the set, and let it be zero if $i>n$. Let $\ut{k}: \R^n \ra \R$ be the function defined as $\ut{k}(S)=\ut{k}(x_1,x_2\ldots x_n) = \u(\sum_{i=1}^k \max_i \{x_1,x_2\ldots x_n\})$, {\em i.e.}  utility of the sum of the $k$ largest arguments.
Let $\u(S)$ denote $\ut{\infty}(S) = \ut{|S|}(S)$, the utility of the sum of all variables. We note an important property of $\ut{k}$ in the following lemma; its proof is deferred to the appendix.
\begin{lemma}\label{lem:submodular}
For any concave utility function $\u$, and any $k$ and $n$, the function $\ut{k}:\R^n\ra \R$ is a symmetric, monotone and submodular.
\end{lemma}
We shall use the following {\em correlation gap} bound established by Agrawal et. al. \cite{ADSY11} for monotone submodular functions.
\begin{lemma}\cite{ADSY11}\label{lem:correlation-gap}
Given $n$ non-negative random variables $X_1, X_2, ..., X_n$ with distributions $D_1, D_2, ..., D_n$, let $D$ be an arbitrary joint distribution over these $n$ random variables such that the marginal distribution for each $X_i$ remains unchanged. Let $D_{ind}$ be the joint distribution where each $X_i$ is sampled from $D_i$ independent of $X_{-i}$. Then for any monotone submodular function $f$ over $X_1, X_2, ..., X_n$, we have  $\frac{{\bf E}_{X\sim D_{ind}} f(X)}{{\bf E}_{X\sim D} f(X)} \ge 1-1/e$.
\end{lemma}
Let $\M_{\OPTs}$ be a utility optimal DSIC mechanism for a $k$-unit auction. It follows that in $\M_{\OPTs}$, every buyer $i$ is offered a (random) price $P_i$ as a function of other buyers' bids; he receives an item and pays the offered price if and only if his value exceeds the price. The following lemma uses the correlation gap to establish the existence of an SPM which is a $(1-1/e)$-approximation to $\M_{\OPTs}$.
\begin{lemma}\label{lem:exists}
Suppose that $\M_{\OPTs}$ offers a (random) price $P_i$ to each buyer $i$ (the prices $P_i,\ 1\le i\le n,$ may be correlated). Let $\M'$ be a randomized SPM that selects an independent random price $P'_i$ for each buyer, such that $P'_i$ and $P_i$ have the same marginal distribution, and offers items to buyers in decreasing order of prices, until the items run out. Then $\u(\M') \ge (1-1/e) \OPT$.
\end{lemma}
\begin{proof}
Let $R_i$ be the payment obtained in $\M_{\OPTs}$ from buyer $i$. Note that $P_i$ and $R_i$ are correlated random variables that depend on the realization of the valuations, and $R_i=P_i$ if $v_i > P_i$, else $R_i=0$. As at most $k$ buyers can make a positive payment in any realization of $\M_{\OPTs}$, we have
\begin{align*}
\u(\M_{\OPTs})=\E{\u(R_1,R_2\ldots R_n)} = \E{\ut{k}(R_1,R_2\ldots R_n)}
\end{align*}
Let $R'_i=P'_i$ if $v_i>P'_i$, else $R'_i=0$. Since the SPM $\M'$ orders buyers in decreasing order of offer prices, so it collects the $k$ largest acceptable prices as payment. We have $\u(\M') = \E{\ut{k}(R'_1,R'_2\ldots R'_n)}$.
Note that $R_i$ and $R'_i$ have the same distribution for each $i$, except that $R_1,R_2\ldots R_n$ are correlated variables, while $R'_1,R'_2\ldots R'_n$ are mutually independent. Using the submodularity of $\ut{k}$ (Lemma \ref{lem:submodular}) and the correlation gap (Lemma  \ref{lem:correlation-gap}), we get
\begin{align*}
\u(\M') = \E{\ut{k}(R'_1,R'_2\ldots R'_n)} \geq (1-1/e) \E{\ut{k}(R_1,R_2\ldots R_n)} = (1-1/e)\u(\M_{\OPTs})
\end{align*}
This completes the proof.
\end{proof}

Correlation gap was used by Yan \cite{Y11} to show the same approximation ratio for an SPM to expected revenue maximization.
However, for revenue maximization, it suffices for the SPM to match a revenue-optimal mechanism only in the probability of sale to each buyer, which solely determines the buyer's contribution to expected revenue. In contrast, for the utility maximization result of Lemma \ref{lem:exists}, {\em the SPM should match a utility-optimal mechanism in the entire distribution of prices to each buyer}. Also, the SPM for revenue maximization is poly-time computable, since a revenue-optimal mechanism is known (Myerson's mechanism).
To the best of our knowledge, the SPM designed in Lemma \ref{lem:exists} is not poly-time computable: constructing it would need an oracle access to a utility-optimal mechanism. Further, as we have to match $\M_{\OPTs}$ for each buyer-price pair, guessing the entire price distribution would require time exponential in the number of buyers.

\subsection{Algorithm to compute a $(1-1/e - \eps)$-approximate SPM}\label{sec:algorithm}
We now present a polynomial time algorithm to compute an SPM whose approximation guarantee essentially matches the existential result above. 
To simplify the exposition of our algorithm, we assume that prices offered by any truthful mechanism belong to some known set $\P=\{p_1, p_2, \ldots \}$ whose size is polynomial in $n$.
Let $\pi_{ij}$ be the probability that buyer $i$ is offered price $p_j$ in $\M_{\OPTs}$.
We divide the prices in $\P$ into $3$ classes, {\em small}, {\em large} and {\em huge}. Fix some $1>\eps>0$. Let $\P_{\hg}$ be the set of {\em huge prices} defined as $p_j\ge \u^{-1}(\OPT/\eps)$. The distinction between small and large prices depend more intricately on the optimal mechanism. Let $p^*$ be the largest price such that $\sum_{\u^{-1}(\OPT/\eps) > p_j\ge p^*} q_j \ge 1/\eps^4$, {\em i.e.} the threshold where the total sale probability of all large prices add up to at least $1/\eps^4$. If such a threshold does not exist, then let $p^*=0$ (note that $p*$ must be zero if $k<1/\eps^4$). Let $\P_{\sm}$ be all prices less than $p^*$, so that $\P_{\lrg}=\{p_j \mid \u^{-1}(\OPT/\eps) > p_j\ge p^*\}$.

In the following lemma, we present a key set of sufficient conditions for a $(1-1/e-\eps)$-approximate mechanism which forms the basis of our algorithm; we defer its proof to later in the section.
\begin{lemma}\label{lem:mass-variant}
Consider any SPM $\M'$, that offers price $p_j$ to buyer $i$ w.p. $\pi'_{ij}$, such that (a) $\sum_{i,p_j\in \P_{\sm}} p_j \pi'_{ij} (1-F_{i}(p_j)) = \sum_{p_j\in \P_{\sm}} p_j q_j$, (b) for each $p_j\in\P_{\lrg}$, $\sum_{i} \pi'_{ij} (1-F_{i}(p_j))=q_j$, (c) $\sum_{i,p_j \in \P_{\hg}} \u(p_j) \pi'_{ij} (1-F_{i}(p_j)) =\sum_{p_j\in \P_{\hg}} \u(p_j) q_j$ and (d) $\sum_{i,p_j\in \P} \pi'_{ij} (1-F_{i}(p_j)) \le k$. Then we have $\u(\M')\ge (1-\frac{1}{e}-O(\epsilon))\OPT$.
\end{lemma}
Lemma~\ref{lem:mass-variant} states that instead of matching $\M_{\OPTs}$ in the probability mass of each $<$buyer, large-price$>$ pair, {\em it suffices to match the total probability mass at each large price, summed over all buyers}. Thus the {\em probability mass can be redistributed across buyers} without much loss in utility.

Further, Lemma \ref{lem:mass-variant} effectively states that the contribution of the small prices and the large prices can be {\em linearized}. Intuitively, if the small prices make a significant contribution to utility, then the mechanism must be collecting many small prices, so the total revenue from small prices exhibits a concentration around its expectation. Moreover, whenever a huge price is obtained in a realization, we can neglect the contribution from all other buyers in that realization, without losing much of the expected utility. So the contribution of huge prices can be measured separately. {\em This separation of huge and small prices from large prices enables us to keep the number of distinct large prices to at most a constant.}

\smallskip
\noindent
{\bf Algorithm:} We give an outline of the algorithm; the details are provided in the appendix. From Lemma~\ref{lem:mass-variant}, it suffices to match $\M_{\OPTs}$ in (a) the expected revenue from the small prices ($\Rev$), (b) the expected contribution to utility from the huge prices $H$, and (c)  the total sale probability at each large price ($q_j$). The values of  these parameters define a {\em configuration}, and we guess the value of each parameter with appropriate discretization. The number of distinct configurations is bounded by $2^{poly(1/\eps)}$. For each configuration, we check if there exists an SPM satisfying the configuration, using the covering linear program (LP) below. In the LP, the variable $x_{ij}$ denotes the probability that buyer $i$ is offered price $p_j$.
\[ \begin{array}{rcll}
\sum_{i}\left(1-F_{i}(p_j)\right)x_{ij} &\ge& q_j  &\forall p_j\in \P_{\lrg}\\
\sum_{i,p_j\in \P_{\sm}} \left(1-F_{i}(p_j)\right) p_j x_{ij} &\ge& \Rev  \\
\sum_{i,p_j\in \P_{\hg}} \left(1-F_{i}(p_{j})\right)\u(p_j) x_{ij} &\ge& H \\
\sum_{i,j} \left(1-F_{i}(p_j)\right)x_{ij} &\le& k \\
\sum_{j} x_{ij}  &\le& 1 &\forall j \\
x_{ij} &\in& [0,1] &\forall i,j
\end{array}\]
Any feasible solution to this linear program gives a distribution of prices for each buyer, which gives us an SPM that satisfies the guessed configuration. We iterate through all the configurations, and pick the best among these SPMs. A deterministic SPM with desired utility guarantees can be easily identified by sampling from this randomized SPM.

\subsection{Proof of Lemma~\ref{lem:mass-variant}}
We begin by introducing two operations on random variables, {\em split} and {\em merge}. Using these two operations, we prove two key properties of random variables in Lemmas~\ref{lem:mass-general} and \ref{lem:new}; Lemma~\ref{lem:mass-variant} would follow as a corollary of these two lemmas.

\smallskip
\noindent
{\bf Split and Merge Operations: } We now define two operations, {\em merge} and {\em split}, on non-negative random variables. In the {\em merge} operation, given a set $S$ of independent non-negative random variables, let $X_i,X_j$ be any two variables in $S$ such that $\Pr{X_i\neq 0}+ \Pr{X_j\neq 0}\le 1$, then variables $X_i, X_j$ are replaced by a new variable $Y$ such that, for each $p> 0$, $\Pr{Y=p}=\Pr{X_i=p}+\Pr{X_j=p}$ and $Y$  is independent of other variables in $S\backslash\{X_i, X_j\}$.

The {\em split} operation breaks a random variable into a set of independent variables.
Formally, given a set $S$ of non-negative (possibly correlated) random variables, first the variables in $S$ are made mutually independent, and then each variable $X_i\in S$ is {\em split} into an arbitrary pre-specified set of independent random variables $\{X_{i1},X_{i2},..., X_{it}\}$ such that for each $p>0$, $\sum_{1\le j\le t} \Pr{X_{ij}=p}=\Pr{X_{i}=p}$ and the sets of variables created are also made mutually independent. Intuitively, the merge operation introduces negative correlation. Analogously, the split operation introduces independence.
In Lemmas~\ref{lem:merge} and~\ref{lem:split}, we establish  useful properties of merge and the split operations for a concave non-decreasing function; their proofs are deferred to the appendix.

\begin{lemma}\label{lem:merge}
Let $S$ be a set of independent non-negative random variables, Let $X_1,X_2 \in S$, and let  $Y$ be the variable formed by merging $X_1$ and $X_2$. Then $\E{\ut{k}(S)}\le \E{\ut{k}((S\setminus\{X_1,X_2\})\cup\{Y\})}$.
\end{lemma}
\begin{lemma}\label{lem:split}
Consider a sequence of split operations on a set $S$ of arbitrarily correlated non-negative random variables and let $S'$ be the set of independent random variables at the end of the split operation. Then $\E{\ut{k}(S')}\ge (1-\frac{1}{e})\E{\ut{k}(S)}$.
\end{lemma}
Using these two operations, we establish an important property in the following lemma, that not only introduces independence across correlated random variables, but also allows to redistribute the probability mass across variables.
\begin{lemma}\label{lem:mass-general}
Given an arbitrarily correlated set $S=\{X_1, X_2, ..., X_n\}$of non-negative random variables, consider any set $S'=\{X_1', X_2', X_3', ..., X_m'\}$ of independent non-negative random variables, such that for each value $p_j>0$, we have $\sum_i\Pr{X_i=p_j}=\sum_i\Pr{X_i'=p_j}$. Then for any concave function $\u$ and any $k>0$, we have $\E{\ut{k}(S')}\ge (1-\frac{1}{e})\E{\ut{k}(S)}$.
\end{lemma}
\begin{proof}
We perform the {\em split} operation on $S$ to create a set $Y=\{Y_{ijl}\}$ of variables as follows: for each $1\leq i\leq n$ and $p_j>0$, create $L\rightarrow \infty$ variables  $\{Y_{ijl}|1\le l\le L\}$ where $Y_{ijl}$ takes value $p_j$ w.p. $\frac{\Pr{X_i=p_j}}{L}$ and $0$ otherwise. Using Lemma~\ref{lem:split}, we get that $\E{\ut{k}(Y)}\ge (1-\frac{1}{e})\E{\ut{k}(S)}$

Now we perform merge operation repeatedly on variables in $Y$ to simulate variables in $S'$. The condition in the lemma statement ensures that such merging is always possible, since $L\rightarrow \infty$. Then by Lemma~\ref{lem:merge}, we get $\E{\ut{k}(S')}\ge \E{\ut{k}(Y)} \ge (1-1/e)\E{\ut{k}(S)}$.
\end{proof}

The following lemma effectively states that, given a set of independent random variables, the contribution to the utility of huge values can be separated, and for small values, the variables can be replaced by their expectations; we defer its proof to the appendix.
\begin{lemma}\label{lem:new}
Given any $\epsilon>0$ and a set of independent non-negative random variables $S=\{X_1, X_2, X_3 \ldots\}$ such that $X_i$ takes value $p_i$ w.p. $\pi_i$ and $0$ otherwise, where $p_1 \ge p_2 \ge p_3 \ge \ldots \ge 0$. Also, suppose that $\sum_{X_i\in S} \pi_i \le k$. Let $\hat{p}$ be a price that satisfies $\hat{p} \ge \u^{-1}\left(\E{\ut{k}(S)}/\eps\right)$, and let
$p^*$ be any price such that  $\sum_{p_i\in [p^*,\hat{p})} \pi_i > \frac{1}{\eps^4}$ ($p^*$ is $0$, if no such price exists). Also, let $S_{\sm} = \{X_i | p_i < p^*\}$, $S_{\lrg} = \{X_i | p^*\le p_i < \hat{p}\}$, and $S_{\hg} = \{X_i | p_i > p^*\}$. Then $\E{\ut{k}(S)}$ is approximated to within a factor of $(1\pm O(\eps))$ by $\sum_{X_i \in S_{\hg}} \E{\u(X_i)}  + \E{\u\left(\E{S_{\sm}} + \sum_{X_i\in S_{\lrg}} X_i\right)}$.
\end{lemma}
\junk{
\begin{proof}
We here give an outline of the proof; the detailed proof is given in the appendix.  The proof consist of 3 stages.
\begin{description}
\item The first step is to separate the contribution of $S_{\hg}$ from the rest. We show that $\E{\ut{k}(S)}$ is approximated within a factor of $(1\pm O(\eps))$ by $\E{\ut{k}(S_{\sm} \cup S_{\lrg})} + \sum_{X_i \in S_{\hg}} \E{\u(X_i)}$.
\item The second step establishes that for $k>1/\eps^4$, one may simply assume that $k=\infty$. In particular, let $S' = S_{\sm} \cup S_{\lrg}$. If $k> 1/\eps^4$, then $(1+\eps) \E{\ut{k}(S')} \ge \E{\ut{k(1+\eps)}(S')} \ge (1-\eps) \E{\u(S')}$.
\item Let $S' = S_{\sm} \cup S_{\lrg}$. Then $\E{\ut{k}(S')}$ is approximated within a factor of $(1\pm O(\eps))$ by $\E{\u\left(\E{S_{\sm}} + \sum_{X_i\in S_{\lrg}} X_i\right)}$. This claim, combined with the first claim, completes the proof of Lemma \ref{lem:new}.
\end{description}
\end{proof}
}
\smallskip
Now we are ready to prove Lemma~\ref{lem:mass-variant}.
The revenue from a buyer in a mechanism can be represented by a random variable, possibly correlated with other buyers' random variables. Let $\M'$ be a mechanism that matches $\M_{\OPTs}$ on the total sale probability for each price, and its sale probability for each $<$buyer, large-price$>$ pair is same as $\M$. Using Lemma~\ref{lem:mass-general}, we get $\u(\M')\ge (1-1/e)\u(\M_{\OPTs})$. As $\M$ and $\M'$ have (approximately) identical revenues from small prices and  utilities from huge prices, we can invoke Lemma~\ref{lem:new} to establish that $\u(\M) \ge (1-\epsilon)\u(\M')$. This completes the proof.

\section{Risk Averse Seller and Risk Averse Buyers: BIC Mechanisms}\label{sec:bayes-nash}
We design mechanisms when buyers as well as the seller are risk averse, and the seller's objective is to maximize his expected utility. We make two assumptions: (a) for any $i,t<0$, $\U_i(0)=0$ and $\U_i(t)=-\infty$; this implies that the mechanism is {\em ex post individually rational}, and the payment from a buyer is $0$ whenever he does not get the item, (b) we further restrict to the set of mechanisms in which payments are always non-negative, i.e. there is no {\em positive transfer} from the seller to a buyer.
Consider any mechanism $\M$ that satisfies these assumptions, let $g(\cdot, \cdot, \cdot)$ and $h(\cdot, \cdot, \cdot)$ be its payment functions. Then we have $g(i,j,v)=0$ for each $i,v$ and payment $p_j$, and $h(i,j,v)=0$ for each $i,v$ and payment $p_j<0$. Thus, the function $g(\cdot, \cdot,\cdot)$ is not required to describe the mechanism. The following theorem summarizes our result.
\begin{theorem}\label{thm:non-iid}
There exists a polynomial time algorithm to compute a truthful-in-expectation mechanism for a $k$-unit auction with expected utility at least $\left(1-\frac{1}{e}\right)^2 \gamma(k) \OPT$ where $\OPT$ is the expected utility of an optimal BIC mechanism.
Moreover, for $k\ge 1/\epsilon^3$, there is a $\left(1-\frac{1}{e}-\epsilon\right)$-approximation.
\end{theorem}
Our result giving an improved approximation ratio of $\left(1-\frac{1}{e}\right)\gamma(k)$ for IID buyers is deferred to Appendix~\ref{sec:iid}. In the rest of the section, we prove Theorem~\ref{thm:non-iid}. Let $\M_{\OPTs}$ be a utility optimal BIC mechanism with these two restrictions, and $h_{\OPTs}(\cdot,\cdot,\cdot)$ be its payment function.

\smallskip
\noindent
{\bf Overall Idea:} Consider any {\em soft} randomized sequential mechanism $\M'$ that processes buyers independently (according to its payment function $h'(\cdot,\cdot,\cdot)$), and matches $\M_{\OPTs}$ for (a) the total probability of each large payment summed over all buyers, (b) the total revenue from small payments and (c) the utility from huge payments. Using techniques developed for DSIC mechanisms, we get $\u(\M')\ge (1-1/e)\OPT$. However, converting such soft mechanism into a mechanism that strictly satisfies the inventory constraint while maintaining truthfulness is not easy. In the case of DSIC mechanisms, the buyers were arranged in a decreasing order of prices, noting that {\em top-$k$} is a sub-modular function. Here, if we allocate items to buyers with {\em top-$k$} payments in a realization of $\M'$, then the mechanism is no longer truthful. Further, as {\em first-$k$} is not a sub-modular function, the desired approximation guarantee cannot be proven if we process buyers according to a fixed order. We get around this problem by constructing a mechanism with $L\rightarrow\infty$ rounds, where in every round, each buyer is processed independently w.p. $1/L$. The revenue from each allocation in this mechanism has an identical distribution. This helps to limit the loss caused by imposing strict inventory constraints.  We now describe our mechanism in detail.

\smallskip
\noindent
{\bf The Mechanism:} Our mechanism $\M_{\rounds}$ consists of $L\rightarrow\infty$ rounds and $h_{\rounds}(\cdot, \cdot,\cdot)$ is the payment function associated with it.  In each round, buyers arrive according to a predefined order. When buyer $i$ arrives, subject to availability of items,  he is independently processed with probability $\frac{1}{L}$ as follows: if his reported valuation is $v$, then he is given an item w.p. $\sum_j h_{\rounds}(i,j,v)$, and whenever he is given an item, he makes a payment of $p_j$ w.p. $\frac{h_{\rounds}(i,j,v)}{\sum_l h_{\rounds}(i,l,v)}$. Once processed, buyer $i$ is not considered for any future rounds. Further, the payment function $h_{\rounds}(\cdot, \cdot, \cdot)$ satisfies following properties:

\smallskip
\noindent
(a) $\sum_{i,v,p_j\in \P_{\sm}} p_j h_{\rounds}(i,j,v)f_i(v) = \sum_{i,v,p_j\in \P_{\sm}} p_jh_{\OPTs}(i,j,v)f_i(v) $, \\
(b) for each $p_j\in\P_{\lrg}$, $\sum_{i,v} h_{\rounds}(i,j,v)f_i(v)= \sum_{i,v} h_{\OPTs}(i,j,v)f_i(v)$, \\
(c) $\sum_{i,v,p_j \in \P_{\hg}} \u(p_j)h_{\rounds}(i,j,v)f_i(v) =\sum_{i,v,p_j \in \P_{\hg}} \u(p_j)h_{\OPTs}(i,j,v)f_i(v)$, \\
(d) for each $i,v, v'$, $\sum_{j}\U_i(v-p_j)h_{\rounds}(i,j,v) \ge \sum_{j} \U_i(v-p_j)h_{\rounds}(i,j,v')$, and \\
(e) $\sum_{i,j,v}h_{\rounds}(i,j,v)f_i(v)\le k$.
\smallskip
\junk{
\begin{algorithm}
\caption{Mechanism $\M_{\rounds}$}\label{alg:bic}.
\begin{algorithmic}
\STATE Choose $L\rightarrow\infty$; the mechanism consists of $L$ rounds.
\FOR{round $l=1:L$}
\FOR{each buyer $i=1:n$}
\STATE Skip buyer $i$ w.p. $\left(1-\frac{1}{L}\right)$.
\IF{buyer $i$ is not skipped in the current round, he is not processed in previous rounds, and the items are still available}
\STATE Let $v$ be his valuation. Allocate the item to him w.p. $\sum_j h_{\rounds}(i,j,v)$ for a payment of $p_j$.
\ENDIF
\ENDFOR
\ENDFOR
\end{algorithmic}
\end{algorithm}
}

We draw a parallel between the properties of $h_{\rounds}(\cdot, \cdot, \cdot)$ with the algorithm developed in the case of DSIC mechanisms: the first three properties are equivalent to designing a mechanism that matches $\M_{\OPT}$ in the total probability for each large payment, the expected revenue from small payments and the expected utility from huge payments. The fourth constraint establishes the truthfulness of $\M_{\rounds}$, and the last constraint ensures its feasibility in expectation. We further note that $\M_{\rounds}$ is {\em truthful-in-expectation}: conditioned on processing buyer $i$ in some round, the payment function ensures truthfulness in terms of his expected utility.

The following lemma bounds the utility of $\M_{\rounds}$, we defer its proof to later in the section.
\begin{lemma}\label{lem:non-iid}
As $L\rightarrow\infty$, $\u\left(\M_{\rounds}\right)\ge (1-\epsilon)\left(1-\frac{1}{e}\right)^2\gamma(k)\OPT$.
\end{lemma}
\noindent
{\bf Algorithm:} To construct an algorithm, we guess the total probability for each large payment ($q_j$), the utility from huge payments ($H$) and the revenue from the small payments ($\Rev$). The feasibility of a configuration can be checked using a covering LP; the details of the LP are given in the appendix. There are $2^{\poly(1/\epsilon)}$ configurations, and we select a feasible configuration with maximum expected utility. Further, the number of rounds can be limited to $O(n^2)$ with a small loss in the approximation factor. To establish our result, it remains to prove Lemma~\ref{lem:non-iid}.

\smallskip
\begin{proofof}{Lemma~\ref{lem:non-iid}}
Let $\I_{\copies}$ be an instance of the problem where each buyer is split into $L$ independent copies, the copies of buyer $i$ are $i1, i2, ..., iL$, and the valuation for each copy is drawn independently from $F_i$. Consider a mechanism $\M_{\soft}$ on $\I_{\copies}$ with $L$ iterations. The $l$th copy of every buyer is considered in the $l$th iteration; when buyer $il$ arrives, $\M_{\soft}$ discards him w.p. $(1-1/L)$, otherwise it processes him according to $h_{\rounds}(i, \cdot, \cdot)$.
In the following lemma, we lower bound the utility of $\M_{\soft}$; its proof follows from Lemma~\ref{lem:mass-general} and Lemma~\ref{lem:new}.
\begin{lemma}\label{lem:many_copies}
$\u(\M_{\soft}) \ge (1-1/e-\epsilon)\OPT$.
\end{lemma}
To simplify notation, in the rest of the proof, we refer to the payment function of $\M_{\soft}$ by $h(\cdot,\cdot,\cdot)$. Further, let $k_{\expt}=\sum_{i,j,v}h_{\rounds}(i,j,v)f_i(v)$; note $k_{\expt}\le k$.

Observe that mechanisms $\M_{\rounds}$ and $\M_{\soft}$ are equivalent with two exceptions: (a) hard inventory constraint of $\M_{\rounds}$, and (b) $\M_{\soft}$ can process more than one copy of a buyer in a realization. We first address the issue of the inventory constraint. Using the correlation gap, we get that the expected number of allocations in $\M_{\soft}$ after first $k$ allocations is at most $k_{\expt}/e$. This alone is not sufficient to prove the lemma as $\u$ is not linear.
We note a crucial property of $\M_{\soft}$ in Lemma~\ref{lem:distribution-hard}, it establishes that {\em the revenue from any allocation in $\M_{\soft}$ has an identical distribution}. Let $\D_i$ be the distribution on the revenue from first $i$ allocations in $\M_{\soft}$.
\begin{lemma}\label{lem:distribution-hard}
As $L\rightarrow \infty$, we have $\Prob_{X_{i}\AC \D_i, X_{i-1}\AC \D_{i-1}}\left[(X_{i}-X_{i-1})=p_j\right]= \frac{\sum_{i,v}h(i,j,v)f_i(v)}{k_{\expt}}$.
\end{lemma}
\begin{proof}
Let $Y_l$ be a random variable indicating the revenue made in round $l$. Clearly $Y_l$ and $Y_{l'}$ have the  same distribution for any $l,l'\le L$. Furthermore, as $L\rightarrow\infty$, conditioned on one allocation in an iteration, the probability of an additional allocation in the same iteration is (almost) $0$. Thus we get $\Pr{Y_l=p_j|Y_l\neq 0} = \frac{\sum_{i,v}h(i,j,v)f_i(v)}{k{\expt}}$.
Conditioned on $i$th allocation happening in round $l$, the probability that $i$th allocation has revenue $p_j$, is exactly $\Pr{Y_l=p_j|Y_l\neq 0}$. This proves the lemma.
\end{proof}

Consider a new mechanism $\M_{\hard}$ on $\I_{\copies}$ that is identical to $\M_{\soft}$ with an exception that it stops after $k$ allocations. We now bound its utility.
\begin{lemma}\label{lem:utility-hard}
As $L\rightarrow \infty$, $\u(\M_{\hard}) \ge \gamma(k)\u(\M_{\soft})$.
\end{lemma}
\begin{proof}
As payments are non-negative, using Lemma~\ref{lem:distribution-hard}, we get $\D_i \succeq\D_{j}$ for each $i$ and $j<i$. The contribution to the utility from the $i$th allocation in $\M_{soft}$ is
$$\Ex_{X\AC \D_i, Y\AC \D_{i-1}}\left[\u(X)-\u(Y) \right] = \Ex_{X\AC \D_1, Y\AC \D_{i-1}}\left[\u(X+Y)-\u(Y) \right]$$
Let $Z_i$ denote the above quantity. Using stochastic dominance of $\D_{i}$ over $\D_{j}$ for every $j<i$ and Lemma~\ref{lem:distribution-hard}, we get that $Z_i\le Z_{j}$ for any $j<i$.

As the allocations in $\M_{\soft}$ are independent, and the expected number of allocations in $\M_{\soft}$ after first $k$ allocations can be bounded by $k_{\expt}(1-\gamma(k))$.
Let $r_1, r_2, ..., r_n$ be the probabilities of $1$st, $2$nd, ..., $n$th allocation in $\M_{\soft}$. We have $r_i\ge r_{>i}$, thus we get
\[ \begin{array}{rcl}
\u(\M_{\hard}) = \sum_{1\le i\le k} r_iZ_i \ge \gamma(k)\left(\sum_{1\le i\le k} r_iZ_i +\sum_{k+1\le i\le n} r_i Z_{k}\right) \ge \gamma(k)\left(\sum_{1\le i\le n} r_iZ_i\right) \ge \gamma(k)\u(\M_{\soft})
\end{array}\]
This proves the lemma.
\end{proof}

To bound the utility of $\M_{\rounds}$, we need to address one more issue: $\M_{\hard}$ can process more than one copy of a buyer. Let $\D_{i1}$ be the distribution on the revenue from all copies of first $i$ buyers in $\M_{\soft}$. Let $\D_{i2}$ be the distribution on the revenue from first $i$ buyers in $\M_{\rounds}$. As payments are non-negative, we have $\D_{i1} \succeq\D_{i2}$. Furthermore, for any fixed $i$, for each $l$, the distribution on the revenue from the $l$th allocation among buyer $i$'s copies is same. The expected number of copies of buyer $i$ processed in $\M_{\soft}$ is $1$. Using correlation gap, the expected number of rounds in which buyer $i$ is processed in $\M_{\hard}$ after first processing is $1/e$. Using stochastic dominance of $\D_{i1}$ over $\D_{i2}$, the expected loss in the utility can be bounded by a factor $1/e$. This completes the proof.
\end{proofof}
Now we give an improved result when $k\ge 1/\epsilon^3$. In a soft randomized sequential mechanism with payment function same as $\M_{\rounds}$,  if we discard each buyer independently w.p. $\epsilon$, then w.p. at least $(1-\epsilon)$, we do not run out of items. The $(1-1/e-O(\epsilon))$-approximation follows by the following property of concave functions.
\begin{lemma}\label{lem:prob-red}
Let $X$ be a random variable that takes value between $0$ and $\Rev$ for some
$\Rev > 0$. Then, $\E{\u(\Rev-X)}\geq \left(1 - \frac{\E{X}}{\Rev}\right)\left(\u(\Rev)-\u(0)\right)$ for any non-decreasing concave function $\u$.
\end{lemma}


\newpage
\appendix
\section{Omitted Proofs from Section~\ref{sec:multi-unit}}
\begin{proofof}{Lemma~\ref{lem:submodular}}
It is obvious that $\ut{k}$ is monotone and symmetric. Let $x=(x_1,x_2\ldots x_n)$ and $y=(y_1,y_2\ldots y_n)$, and let $\max \{ x,y \} = (z_1,z_2\ldots z_n)$ and $\min \{ x,y \} = (z'_1,z'_2\ldots z'_n)$, such that $\forall 1\le i\le n$, $z_i = x_i, z'_i = y_i$ if $x_i \ge y_i$, and $z_i = y_i, z'_i = x_i$ if $x_i < y_i$. To show submodularity we need to show that for any $x,y$

$$\ut{k}(\max \{ x,y \}) + \ut{k}(\min \{ x,y \}) \leq \ut{k}(x) + \ut{k}(y)$$
Or equivalently, that for any $x\ge y$ ({\em i.e.} $x_i \ge y_i \forall i$), $1\le i\le n$, and $t \ge x_i\ge y_i$, if $x\wedge_i t= x_1\ldots x_{i-1},t,x_{i+1}\ldots$ and $y \wedge_i t = y_1\ldots y_{i-1},t,y_{i+1}\ldots$, then

$$\ut{k}(x\wedge_i t) - \ut{k}(x) \le \ut{k}(y \wedge_i t) - \ut{k}(y)\ .$$

We shall prove the latter statement. Let $I_k(x) = \sum_{j=1}^k \max_j (x)$. It is easy to see that $I_k(x\wedge_i t) - I_k(x) = \max\left\{0, t - \max\{x_i, \max_k(x)\}\right\} \le \max\left\{0, t - \max\{y_i, \max_k(y)\}\right\} = I_k(y \wedge_i t) - I_k(y)$, {\em i.e.} $I_k$ itself is submodular. Also, $I_k(x)\ge I_k(y)$. The statement now follows from noting that $\ut{k}(x) = \u(I_k(x))$, and that $\u(a+b) - \u(b) \le \u(a'+b') - \u(b')\ \forall a\le a', b\ge b'$, for any concave function $\u$.
\end{proofof}

\begin{proofof}{Lemma~\ref{lem:merge}}
The lemma follows from the fact that $\ut{k}$ is submodular -- in fact, it applies to any submodular function.
Let $\{p_1, p_2, ..., p_{\ell}\}$ be the set of different non-zero values realized by either $X_1$ or $X_2$ with non-zero probability. Let $\pi_i, \pi'_i$ be the probabilities that $X_1$ and $X_2$ realize to $p_i$, respectively. Fix a realization of random variables in $S\backslash\{X_1, X_2\}$, and let this realization be denoted by $\hat{z}$; in this case, we shall express $\ut{k}(S)$ as $\ut{k}(X_1,X_2,\hat{z})$ and $\ut{k}((S\setminus\{X_1,X_2\})\cup\{Y\})$ as $\ut{k}(Y,\hat{z})$.

\[ \begin{array}{l}
\E{\ut{k}\left( Y, \hat{z} \right)} - \E{\ut{k}\left(X_1, X_2 , \hat{z}\right)} \\
=\sum_i (\pi_i+\pi'_i)\ut{k}(p_i)-\sum_i\left(\pi_i(1-\sum_j\pi'_j)+\pi'_i(1-\sum_j\pi_j)\right)\ut{k}(p_i)-\sum_{i,j}\pi_i\pi'_j \ut{k}(p_i+p_j) \\
=\sum_i \left(\pi_i(\sum_j\pi'_j)+\pi'_i(\sum_j\pi_j)\right)\ut{k}(p_i)-\sum_{i,j}\pi_i\pi'_j\ut{k}(p_i+p_j)\\
=\sum_{i,j}\pi_i\pi'_j\left(\ut{k}(p_i)+\ut{k}(p_j)-\ut{k}(p_i+p_j)\right) \ge 0
\end{array}\]
The last inequality follows from the definition of submodularity. Since the above inequality holds for any $\hat{z}$, the lemma follows.
\end{proofof}

\begin{proofof}{Lemma~\ref{lem:split}}
This lemma holds for any monotone submodular function, including $\ut{k}$. If a variable $X_i\in S$ is split into $\{X_{i1}, X_{i2}, ..., X_{it}\}$, then the realization of $X_i$ can be simulated by $\sum_{1\le j\le t}X_{ij}$  with appropriate correlation among them. Thus with appropriate correlation among the variables in $S'$, we can ensure that $\E{\ut{k} (S')}=\E{\ut{k}(S)}$. When variables in $S'$ are made independent, its expected utility is at least $(1-\frac{1}{e})\E{\ut{k}(S)}$,   by the correlation gap bound in Lemma \ref{lem:correlation-gap}, since $\ut{k}$ is monotone and submodular.
\end{proofof}

\begin{proofof}{Lemma~\ref{lem:new}}
Our proof has 3 steps. Our first step is to separate the contribution of $S_{\hg}$ from the rest.
\begin{claim}
$\E{\ut{k}(S)}$ is approximated within a factor of $(1\pm O(\eps))$ by $\E{\ut{k}(S_{\sm} \cup S_{\lrg})} + \sum_{X_i \in S_{\hg}} \E{\u(X_i)}$.
\end{claim}
\begin{proof}
It is easy to upper bound the expected utility: $\E{\ut{k}(S)} \le \E{\ut{k}(S_{\sm} \cup S_{\lrg})} + \E{\ut{k}(S_{\hg})} \le \E{\ut{k}(S_{\sm} \cup S_{\lrg})} + \sum_{X_i \in S_{\hg}} \E{\u(X_i)}$, since $\u$ is concave.

For the lower bound, note that $\Pr{\ut{k}(S_{\hg}) > 0} \le \eps$, by Markov's inequality, since $\ut{k}(S_{\hg})>0$ implies $\ut{k}(S) \ge \E{\ut{k}(S)}/\eps$. Further, the probability that some particular $X_i \in S_{\hg}$ is non-zero, while all other variables in $S_{\hg}$ are zero, is at least $\pi_i (1-\eps)$.
So we have
\begin{equation*}
\begin{aligned}
\E{\ut{k}(S)} &=  \E{\ut{k}(S) - \ut{k}(S_{\hg})} + \E{\ut{k}(S_{\hg})}\\
&\ge  \Pr{\ut{k}(S_{\hg}) = 0} \E{\ut{k}(S_{\sm} \cup S_{\lrg}) \mid \ut{k}(S_{\hg}) = 0} + \E{\ut{k}(S_{\hg})} \\
&\ge (1-\eps) \E{\ut{k}(S_{\sm} \cup S_{\lrg})} +  \sum_{X_i\in S_{\hg}}  \pi_i (1-\eps)\u(p_i)\\
&=   (1-\eps) \E{\ut{k}(S_{\sm} \cup S_{\lrg})} + (1-\eps) \sum_{X_i\in S_{\hg}}  \E{\u(X_i)}
\end{aligned}
\end{equation*}
\end{proof}

\noindent

Without loss of generality, we may assume that $1/\eps$ is an integer. Our second step establishes that for $k>1/\eps^4$, one may simply assume that $k=\infty$.

\begin{claim}
Let $S' = S_{\sm} \cup S_{\lrg}$. If $k> 1/\eps^4$, then $(1+\eps) \E{\ut{k}(S')} \ge \E{\ut{k(1+\eps)}(S')} \ge (1-\eps) \E{\u(S')}$.
\end{claim}
\begin{proof}
The first inequality trivially follows from the definition of $\ut{k}$, since $(1+\eps) \ut{k}(S') \ge \ut{k(1+\eps)}(S')$ on every realization.

Since $k> 1/\eps^4$ and $\sum_{X_i\in S'} \pi_i \le k$, so the probability of the event that the number of non-zero variables in a realization exceeds $k(1+\eps)$, is at most $\eps$. This follows directly from Chebyshev's inequality. Conditioned on this event, the expected contribution to $\u(S')$ from variables that are positive but does not contribute to $\ut{k(1+\eps)}(S')$, is at most $\E{\u(S')}$. This is because the realizations of these {\em excess positive variables} is independent of the said event. So the difference between $\E{\u(S')}$ and $\E{\ut{k(1+\eps)}(S')}$ is at most $\eps \E{\u(S')}$, hence the second inequality.
\end{proof}

Our third claim, combined with the first claim, completes the proof of Lemma \ref{lem:new}.
\begin{claim}
Let $S' = S_{\sm} \cup S_{\lrg}$. Then $\E{\ut{k}(S')}$ is approximated within a factor of $(1\pm O(\eps))$ by $\E{\u\left(\E{S_{\sm}} + \sum_{X_i\in S_{\lrg}} X_i\right)}$.
\end{claim}
\begin{proof}
Also, if $k\le 1/\eps^4$, then $S_{sm}=\emptyset$ by definition, so the claim is trivially true; so we assume that $k>1/\eps^4$. Then, by our second claim above, $\E{\ut{k}(S')}$ is approximated within a factor of $(1 + 2\eps)$ by $\E{\u(S')}$. So it suffices to show that $\E{\u(S')}$ is approximated within a factor of $(1\pm O(\eps))$ by $\E{\u\left(\E{S_{\sm}} + \sum_{X_i\in S_{\lrg}} X_i\right)}$ .
Again, the upper bound is easy to obtain: $\E{\u(S')} \le \E{\u\left(\E{S_{\sm}} + \sum_{X_i\in S_{\lrg}} X_i\right)}$ by concavity of $\u$.

For the lower bound, let us first consider the case that $\E{S_{\sm}} \ge p^*/\eps^3$. Then $\Var{S_{\sm}} \le \eps^3 \E{S_{\sm}}$, since $X_i < p^*$ for all $X_i \in S_{\sm}$. Then by Chebyshev's inequality, we have $\sum_{X_i\in S_{\sm}} X_i > (1-\eps) \E{S_{\sm}}$ with probability at least $1-\eps$. So $$\E{\u(S')}\ge (1-\eps) \E{\u\left((1-\eps)\E{S_{\sm}} + \sum_{X_i\in S_{\lrg}} X_i\right)} \ge (1-2\eps) \E{\u\left(\E{S_{\sm}} + \sum_{X_i\in S_{\lrg}} X_i\right)}\ .$$

Now let us assume that $\E{S_{\sm}} < p^*/\eps^3$. In this case, we show below that
$$\E{\u(S')} \le \E{\u\left(\E{S_{\sm}} + \sum_{X_i\in S_{\lrg}} X_i\right)} \le (1+O(\eps)) \E{\u(S_{\lrg})} \le (1+O(\eps)) \E{\u(S')}\ .$$
Note that with probability at least $(1 - \eps)$, the number of non-zero variables in $S_{\lrg}$ is at least $(1-\eps)/\eps^4$ (by Chebyshev's inequality on the indicator variables of $X_i$), which implies that $\sum_{X_i\in S_{\lrg}} X_i \ge (1-\eps)p^*/\eps^4$ -- let this event be denoted by $A$, so that $\Pr{A}\ge (1-\eps)$, and let $\bar{A}$ denote the complementary event. Let $\E{S_{\sm}} = \Rev_{\sm}$.

\begin{equation*}
\begin{aligned}
&\E{\u\left(\Rev_{\sm} + \sum_{X_i\in S_{\lrg}} X_i\right) - \u(S_{\lrg})}\\
&= \E{\u\left(\Rev_{\sm} + \sum_{X_i\in S_{\lrg}} X_i\right) - \u(S_{\lrg}) \mid A} \Pr{A} + \E{\u\left(\Rev_{\sm} + \sum_{X_i\in S_{\lrg}} X_i \right) - \u(S_{\lrg}) \mid \bar{A}} (1- \Pr{A}) \\
&\le (1-\eps)\left(\u\left(\Rev_{\sm} + \frac{(1-\eps)p^*}{\eps^4}\right) - \u\left(\frac{(1-\eps)p^*}{\eps^4}\right)\right)  +
\epsilon \left(\u(\Rev_{\sm} + 0) - \u(0)\right)  \\
&\le (1-\eps) \frac{\eps^4 \R_{\sm}}{(1-\eps)p^*} \u\left(\frac{(1-\eps)p^*}{\eps^4}\right) + \eps \u(\Rev_{\sm})\\
&\le \eps  \u\left(\frac{(1-\eps)p^*}{\eps^4}\right) + \eps \u\left( \frac{p^*}{\eps^3} \right) \le 2\eps \u\left(\frac{(1-\eps)p^*}{\eps^4}\right)\\
\end{aligned}
\end{equation*}
Both the first and second inequalities follow from concavity of $\u$, while the last inequality holds for $\eps\le \frac{1}{2}$. Finally, observe that $\E{\u(S_{\lrg})} \ge (1-\eps) \u\left(\frac{(1-\eps)p^*}{\eps^4}\right)$.
\end{proof}
\end{proofof}

\subsection{Detailed Algorithm for the DSIC Mechanism}
From Lemma~\ref{lem:mass-variant}, it suffices to guess (a) the threshold $p^*$ separating large and small prices, (b) the total probability of sale $q_k$ for each large price $p_k$, (c) the expected revenue from small prices $\Rev$, and (d) the expected contribution to utility, say $h$, from huge prices. The large prices (and hence the threshold $p^*$ as well) can be discretized so that they come only from the set $\u^{-1} (\eps^5 t \OPT)$ for any integer $0\le t\le \frac{1}{\eps^6}$. This is because the loss in utility due to discretization, in every realization,  is at most $\eps^5 \OPT$ per sale at a large price. This is immediate if the large price sale is the first sale to occur in the realization, and concavity implies that the loss can only be smaller if some revenue has already been collected in the realization. So the total loss in expected utility due to this discretization is at most $\eps^5\OPT$ times the expected number of large prices sales. Since there are $O(\frac{1}{\eps^4})$ large price sales in expectation, the loss is $O(\eps \OPT)$.

Similarly, $h$ and $\u(\Rev)$ can be guessed up to a multiple of $\eps\OPT$. We shall guess the total sale probability for each large price to the nearest multiple of $\epsilon^{8}$. The total error over all price buckets is less  than $\eps^2$ and so the loss in the utility due to rounding of the probability masses is bounded  by $\epsilon\OPT$. All these guesses define a {\em configuration}. The number of configurations is bounded by $2^{poly(1/\eps)}$.
 For each configuration we check if there exists an SPM satisfying the configuration, using the covering linear program (LP) below. In this, the variable $x_{ij}$ denotes the probability that buyer $i$ is offered price $p_j$.
\[ \begin{array}{rcll}
\sum_{i}\left(1-F_{i}(p_j)\right)x_{ij} &\ge& q_j  &\forall p_j\in \P_{\lrg}\\
\sum_{i,p_j\in \P_{\sm}} \left(1-F_{i}(p_j)\right) p_j x_{ij} &\ge& \Rev  \\
\sum_{i,p_j\in \P_{\hg}} \left(1-F_{i}(p_{j})\right)\u(p_j) x_{ij} &\ge& H \\
\sum_{i,j} \left(1-F_{i}(p_j)\right)x_{ij} &\le& k \\
\sum_{j} x_{ij}  &\le& 1 &\forall j \\
x_{ij} &\in& [0,1] &\forall i,j
\end{array}\]
Any feasible solution to this linear program gives a distribution of prices for each buyer, which gives us an SPM that satisfies the guessed configuration. For the configuration that matches $\M_{\OPTs}$, any feasible solution of the program gives an SPM that is a $(1-1/e-\eps)$-approximation. So our algorithm is to iterate through all the configurations, generate at most one randomized SPM for each configuration, and then simulate them to find out, with high confidence, the best among these SPMs.

\section{BIC Mechanisms for a Risk Averse Seller and Risk Averse IID Buyers}\label{sec:iid}
In this section, we consider the problem of designing a BIC mechanism for a risk averse seller when buyers are risk averse, each buyer's valuation is drawn IID from a known distribution $F$ and their utility functions $\U$ are identical. The following theorem summarizes our result in this setting.
\begin{theorem}\label{thm:iid}
There exists a polynomial time algorithm to compute a truthful-in-expectation mechanism with expected utility $(1-\epsilon)\left(1-\frac{1}{e}\right)\gamma(k)\OPT$, where $\OPT$ is the expected utility of the optimal BIC mechanism.
\end{theorem}
In the rest of the section, we prove Theorem~\ref{thm:iid}. We call a mechanism to be symmetric if, for all $i\neq i',j,v$, we have $h(i,j,v)=h(i',j,v)$. We first establish an important property of symmetric mechanisms in the following lemma.
\begin{lemma}\label{lem:symmetric}
There exists a utility optimal BIC mechanism that is symmetric for all buyers.
\end{lemma}
\begin{proof}
If the optimal BIC mechanism $\M_{\OPTs}$ is not symmetric, then consider a new mechanism $\M'$ that for every set of reported bids, renames buyers randomly, and runs $\M_{\OPTs}$ on the these bids. Clearly $\u(\M_{\OPTs})=\u(\M')$. To prove the lemma, it suffices to establish that $\M'$ is a BIC mechanism. Fix any buyer $i$, his utility in $\M'$ for reporting his true valuation $v$ is
$$\frac{1}{n}\left(\sum_{j,l} \left(\U(v-p_j)h(l,j,v)\right)\right)$$
and his utility when he reports $v'$ instead of true valuation $v$ is
$$\frac{1}{n}\left(\sum_{j,l} \left(\U(v-p_j)h(l,j,v')\right)\right)$$
From the incentive compatibility property of $\M_{\OPTs}$, we get
$$\sum_{j} \left(\U(v-p_j)h(l,j,v)\right)\ge \sum_{j} \left(\U(v-p_j)h(l,j,v')\right)$$
for every $l$. Summing over $l$ proves the lemma.
\end{proof}

So from here onwards, we can assume that $\M_{\OPTs}$ is symmetric; in the optimal mechanism, let $h(j,v)$ be the probability that the buyer gets an item for payment $p_j$ when his valuation is $v$. Thus the revenue of $\M_{\OPTs}$ can be represented using the sum of identical (and correlated) random variables $X_1, X_2, ..., X_n$, where for each $i$, $\Prob[X_i=p_j] = \sum_v h(j,v) f(v)$. Let $X'_1, X'_2, ..., X'_n$ be a set of random variables such that the marginal distribution of $X_i$ and $X'_i$ is same for each $i$. Let $\M_{\soft}$ be a soft randomized sequential mechanism whose allocation and payment functions are same as $\M_{\OPTs}$. Using correlation gap, we get
$$\u\left(\M_{\soft}\right) = \E{\u\left( \sum_i X'_i\right)}\ge \left(1-\frac{1}{e}\right) \OPT$$
Let $\M$ be the mechanism obtained by applying the inventory constraint on $\M_{\soft}$, i.e. it stops when items run out. We note that $\M$ is still a BIC mechanism: even though the probability of allocations and payments for a buyer for valuation $v$ is not same in $\M_{\soft}$ and $\M$, these values remain same conditioned on reaching a buyer without running out of inventory.
We now bound the expected utility of $\M$ in the following lemma, its proof is similar to the proof of Lemma~\ref{lem:utility-hard}.
\begin{lemma}
$\u\left(\M\right)\ge \gamma(k)\M_{\soft}$.
\end{lemma}

Thus it suffices to guess the values $h(j,v)$ for each $j,v$ to compute a $(1-1/e)^2$-approximate mechanism. However such algorithm will require exponential time. We need the following lemma to get a polynomial time algorithm.
\begin{lemma}\label{lem:mass-variant-iid}
Consider two randomized and symmetric sequential pricing mechanism $\M_1$ and $\M_2$ with inventory constraint of $k$, and let $h_1(v,j)$ and  $h_2(v,j)$ be their payment functions, such that (a) $\sum_{v,p_j\in \P_{\sm}} p_j h_1(j,v)f(v) = \sum_{v,p_j\in \P_{\sm}} p_j h_2(j,v)f(v) $, (b) for each $p_j\in\P_{\lrg}$, $\sum_{v} h_1(j,v)f(v)= \sum_{v} h_2(j,v)f(v)$, and (c) $\sum_{v,p_j \in \P_{\hg}} \u(p_j)h_1(j,v)f(v) =\sum_{v,p_j \in \P_{\hg}} \u(p_j)h_2(j,v)f(v)$.
\end{lemma}
Its proof is similar to the proof of Lemma~\ref{lem:mass-variant}. We note an important step in the proof: if two symmetric mechanisms have the same probability for each payment, then the revenue distribution per buyer remains same, and we do not lose a factor of $(1-1/e)$ in the approximation ratio during the split-and-merge operation. 

Thus per buyer values of probability for each large payment ($q_j$), the utility from huge payments ($H$) and the revenue from the small payments ($\Rev$); the feasibility of a configuration can be checked using an LP.

\section{Omitted Details from Section~\ref{sec:bayes-nash}}
\paragraph{BIC Mechanisms for a Risk Averse Sellers: LP to check the feasibility of a configuration}

\[ \begin{array}{rcll}
\sum_{j}\U_i(v-p_j)h_{\rounds}(i,j,v) &\ge& \sum_{j} \U_i(v-p_j)h_{\rounds}(i,j,v') &\forall i,v,v'\\
\sum_{i,v} h_{\rounds}(i,j,v)f_i(v) &\ge& q_j &\forall p_j\in \P_{\lrg}\\
\sum_{v,p_j\in \P_{\sm}} p_j f_i(v)h_{\rounds}(i,j,v) &\ge& \Rev  \\
\sum_{v,p_j\in \P_{\hg}} \u(p_j)f_i(v) h_{\rounds}(i,j,v) &\ge& H \\
\sum_{i,j,v}h_{\rounds}(i,j,v)f_i(v) &\le& k \\
\sum_{j} h_{\rounds}(i,j,v) &=& 1 &\forall i,v\\
h_{\rounds}(i,j,v) &\in& [0,1] &\forall i,j,v \\
\end{array}\]

\section{Mechanisms for a Risk Neutral Seller and Risk Averse Buyers}\label{sec:risk-neutral-seller}
In this section, we design mechanisms for a risk neutral seller when buyers are risk averse. The following theorem summarizes our result; it shows that (a) optimal BIC mechanism can be designed for risk averse seller and (b) the gap between an optimal truthful-in-expectation and an optimal BIC mechanisms is bounded by $\gamma(k)$.
\begin{theorem}\label{thm:bic_risk_neutral}
Given risk averse buyers and a risk neutral seller, let $\OPT$ be the expected revenue of the optimal BIC mechanism. Then (a) there exists a polynomial time algorithm to compute a truthful-in-expectation mechanism with expected revenue $\gamma(k)\OPT$, and (b) there exists a polynomial time algorithm to compute a BIC mechanism with expected revenue $\OPT$.
\end{theorem}
\begin{proof}
We begin by proving the first part of the theorem.
We note that the objective of the following LP is at least $\OPT$. The first constraint in the LP ensures that the mechanism is truthful, the second constraint ensures that exactly one allocation-payment event realizes in expectation for a buyer's bid, and the third constraint is the inventory constraint.
\[ \mbox{Maximize } \sum_{i,j,v} \left(g(i,j,v)p_jf_i(v)+h(i,j,v)p_jf_i(v)\right)\]
\[ \begin{array}{rcll}
\sum_{j}\left(\U_i(v-p_j)h(i,j,v) + \U_i(-p_j)g(i,j,v)\right) &\ge& \sum_{j} \left(\U_i(v-p_j)h(i,j,v') + \U_i(-p_j)g(i,j,v')\right) &\forall i,v,v'\\
\sum_{j}\left(h(i,j,v) +g(i,j,v)\right) &=& 1 &\forall i,v \\
\sum_{i,j,v} h(i,j,v)f_i(v) &\le& k\\
h(i,j,v), g(i,j,v) &\in& [0,1] &\forall i,j,v \\
\end{array}\]
We construct a mechanism from the LP solution using Algorithm~\ref{alg:bic}.
\begin{algorithm}
\caption{Truthful in Expectation Mechanism for a Risk Neutral Seller and Risk Averse Buyers}\label{alg:bic}.
\begin{algorithmic}
\STATE Arrange the buyers in the decreasing order of $\frac{\sum_{j,v} \left(g(i,j,v)p_jf_i(v)+h(i,j,v)p_jf_i(v)\right)}{\sum_{j,v} \left(g(i,j,v)f_i(v)+h(i,j,v)f_i(v)\right)}$.
\FOR{each buyer $i$ in the sequence}
\IF{there is at least one item remaining}
\STATE Ask buyer $i$ for his valuation. Let $v$ be his reported valuation.
\STATE Allocated an item to him w.p. $\sum_{j} h(i,j,v)$, and he makes a payment of $p_j$ w.p.  $\frac{h(i,j,v)}{\sum_{l} h(i,l,v)}$ in this case. When he is not given an item, he makes a payment of $p_j$ w.p. $\frac{g(i,j,v)}{\sum_{l} g(i,l,v)}$.
\ENDIF
\ENDFOR
\end{algorithmic}
\end{algorithm}
We note that the mechanism is truthful-in-expectation: the position of buyer $i$ in the sequence is independent of his valuation, and conditioned on reaching him without running out of items, he allocation and payment probabilities are equal to the LP solution. We now establish the approximation bound.

For the purpose of analysis, we replace the payment from buyer $i$ with a deterministic quantity ${\tt payment_i}=\sum_{j,v} \left(g(i,j,v)p_jf_i(v)+h(i,j,v)p_jf_i(v)\right)$. Let ${\tt prob_i}=\sum_{j,v} \left(g(i,j,v)f_i(v)+h(i,j,v)f_i(v)\right)$, i.e. the probability of allocation to buyer $i$. Further, let $k_{\expt}=\sum_{i,j,v} h(i,j,v)f_i(v)$. W.l.o.g., assume that the mechanism processes buyers in order $1,2,3,..., n$. Let $r_i$ be the probability that $k$ items are allocated before buyer $i$ arrives. We have $r_1\le r_2\le ....\le r_n$, and
$$\sum_i r_i\times {\tt prob_i} \le (1-\gamma(k))\times k$$
Let $t_i=\frac{{\tt payment_i}}{{\tt prob_i}}$. We get $t_1\ge t_2\ge ... \ge t_n$. The revenue of the mechanism is
\begin{align*}
&\sum_i (1-r_i)\times {\tt payment_i}\\
&=\sum_i(1-r_i)\times  t_i\times {\tt prob_i}\\
&\ge \gamma(k)\sum_i {\tt payment_i}
\end{align*}
where the last inequality follows using $t_1\ge t_2\ge ... \ge t_n$, and $\sum_i r_i\times {\tt prob_i} \le (1-\gamma(k))\times k$. This completes the proof of the first claim.

Now we consider the second part of the theorem. Its proof is similar to the first part; we consider an LP as earlier, however we also add a separation oracle for the Border's inequality; \cite{AFHHM12} show that such separation oracle can be designed in polynomial time. The constructive proof of the Border's inequality also gives a way of allocating items to buyers based on their valuations honouring supply constraints. When an item is allocated to buyer $i$, then he is charged $p_j$ w.p. $\frac{h(i,j,v)}{\sum_j h(i,j,v)}$. Similarly, when he is not allocated an item, he pays $p_j$ w.p. $\frac{g(i,j,v)}{\sum_j g(i,j,v)}$. Thus the payment functions of a buyer can be implemented exactly.
\end{proof}

\section{DSIC Mechanisms for Risk Averse Sellers: Unit Demand Buyers}\label{sec:unit-demand}
In this section, we consider the problem of designing DSIC mechanisms for a risk averse sellers in a multi-parameter settings where buyers are unit-demand. We are able to combine our techniques with a recent approximation result for revenue maximization for the unit-demand problem, when there are multiple distinct items. We note that for such multi-parameter settings, no characterization of an optimal mechanism is known even for the expected revenue objective. This gives our second result.
\begin{theorem}\label{thm:unit-demand}
For unit-demand setting with multiple distinct items, there is a poly-time computable randomized order-independent sequential posted price mechanism with expected utility at least $\left(\frac{1-1/e}{6.75} - \eps \right) \OPT$, for any $\eps>0$, where $\OPT$ is the expected utility of an optimal DSIC mechanism.
\end{theorem}

Our obtain our result by combining our tools for handling risk aversion with a recent result of Chawla et. al.~\cite{CHMS10}, who gave a $\frac{1}{6.75}$-approximation to this problem for the expected revenue objective. The approximation factor has since been improved to $\frac{1}{4}$ for expected revenue, by Alaei~\cite{A11}. To the best of our knowledge, our algorithm gives the first non-trivial approximation guarantee for maximizing expected (concave) utility in any multi-parameter setting.

In the rest of the section, we prove Theorem~\ref{thm:unit-demand}.
We will reduce the problem of designing a mechanism for any instance $\I$ of a unit-demand setting with $n$ buyers and $m$ items to finding a special type of mechanism for an instance $\I'$ of a single parameter setting with {\em matrix feasibility constraint}. Matrix feasibility constraint is defined as follows: there are $mn$ buyers arranged in an $n\times m$ matrix, and the seller can serve at most one buyer from each row and at most one buyer from each column.
There is a natural mapping between the two instances: split buyer $i\in \I$ into $m$ independent buyers $i1, i2, ..., im$ and the distribution on the valuation of buyer $ij$ is same as the distribution of buyer $i$'s valuation for item $j$. Clearly any feasible assignment of items in $\I$ corresponds to a feasible assignment of items in $\I'$ and vice versa.

We define a class of pricing mechanisms, called {\em partially order-oblivious posted pricing (partial-OPM)}, for the single parameter setting with matrix constraint: the mechanism specifies an ordering $\sigma$ on the rows on the matrix, along with independent price distributions for all buyers. All buyers in row $\sigma(i)$ ($i^{th}$ row in the sequence) of the matrix arrive before any buyer in row $\sigma(i+1)$; however, the order of buyers in each row $\sigma(i)$ is  chosen adversarially, where the adversary {\em knows the valuation of buyers and the mechanism's prices and wishes to minimize the utility obtained}. Partial-OPM in $\I'$ can be mapped naturally to an SPM for $\I$: buyers arrive in the order specified by $\sigma$, and if buyer $ij$ was offered a price of $x$ in $\I'$, then buyer $i$ is offered item $j$ at a price $x$ (if item $j$ is available when buyer $i$ arrives). The following lemma describes the relation between these two types of problems.

\begin{lemma}\label{thm:relation}
Given a unit-demand instance $\I$ and its corresponding single-parameter instance $\I'$ as defined above, if $\OPT$ and $\OPT'$ are  the optimal expected utilities for $\I$ and $\I'$ respectively, then $\OPT' \ge \OPT$. Furthermore, given any partial-OPM for $\I'$, there is a poly-time computable SPM for $\I$ with expected utility
at least as much as the expected utility of the partial-OPM for $\I'$.
\end{lemma}
\begin{proof}
The scheme for designing an approximate utility-optimal mechanism mirrors the scheme for designing approximately revenue-optimal mechanism in \cite{CHMS10}, as follows.
Fix any utility optimal deterministic mechanism $\M_{\OPTs}$ for $\I$. Note that, any deterministic truthful multi-parameter mechanism for unit demand buyers can be interpreted as offering a price menu with one price for each item/service to each buyer as a function of other buyers' bids \cite{W97}. We construct a corresponding mechanism $\M'$ on $\I'$ that has the same allocation function as in \cite{CHMS10}.  In {\em every realization of the bid vector}, the revenue in $\M'$ is {\em at least as large as} the revenue in $\M_{\OPTs}$. Thus the utility of the former is at least the utility of the latter. So the optimal expected utility on $I'$ is at least as
much as the optimal expected utility on $\I$.

Next, similar to~\cite{CHMS10}, we show that we can use a partial-OPM $\M'$ for $\I'$ to construct an item pricing mechanism $\M$ for $\I$. Order the buyers in $\I$ as per the order of rows of matrix in $\M'$, and fix a realization of buyers' valuations  and prices in $\M'$. In the corresponding realization in $\I$, when buyer $i\in \I$ arrives, he is shown the set of available items and the price asked for item $j$ is the price asked for buyer $ij\in\I'$. As the buyer $i$ chooses the item that maximizes his utility, the revenue from buyer $i$ is no less than the revenue from buyers in row $i$ in $\I'$ (since the mechanism in $\I'$ is order-oblivious within a row). Thus the revenue obtained in $\M$ is at least the minimum revenue (among all orderings of buyers within each row) in $\M'$, for {\em every realization of prices and valuations}. Hence the expected utility of $\M$ is at least as much as the expected utility of $\M'$.This completes the proof of the lemma.
\end{proof}

It is worth noting that Lemma \ref{thm:relation} deviates from the reduction used by Chawla et. al. \cite{CHMS10} from unit-demand setting to a single-parameter setting for expected revenue in the following manner: Chawla et. al. \cite{CHMS10} computed a {\em completely order-oblivious} posted pricing (where even buyers in different rows can be intermingled adversarially) for $\I'$, which yields a {\em completely order-oblivious} posted pricing for $\I$. This implication remains true for utility as well. However, we are not able to analyze the performance of a complete OPM in the case of expected utility; instead, we are able to compute a partial-OPM with a good approximation guarantee on expected utility in $\I'$. As will be clear below, the only part where we need to restrict ourselves to partial-OPMs is Lemma \ref{lem:profit-opm}. Its analogue for revenue holds even for complete OPMs by a simple application of Markov's inequality \cite{CHMS10}, but the proof is more complicated for utility maximization.

In the rest of the section, we prove the following theorem, which combined with Lemma~\ref{thm:relation} implies Theorem~\ref{thm:unit-demand}.

\begin{theorem}\label{thm:matrix}
For the single parameter setting with matrix feasibility constraint, there is a polynomial time algorithm to compute a partial-OPM with a given ordering among the rows and expected utility $\left(\frac{1-1/e}{6.75} - \eps \right)\OPT$. Furthermore, the price distributions computed by the algorithm yields this guarantee for any ordering $\sigma$ among the rows.
\end{theorem}

\noindent
Let $\M_{\OPTs}$ be an optimal deterministic truthful mechanism for a single-parameter instance with matrix constraint, so by incentive compatibility, every buyer is offered a price as a function of other buyers' bids. Let $\P=\{p_1, p_2, \ldots\}$ be the set of distinct prices offered in any mechanism. Let $\pi_{ijk}$ be the probability that buyer $ij$ is offered price $p_k$ in $\M_{OPTs}$.
We note the important properties of prices in $\M_{\OPTs}$: (a) $\sum_{j,k} \pi_{ijk}\left(1-F_{ij}(p_k)\right) \le 1\ \forall i$ (at most one buyer is served from the $i^{th}$ row), (b) $\sum_{i,k} \pi_{ijk}\left(1-F_{ij}(p_k)\right) \le 1\ \forall j$ (at most one buyer is served from the $j^{th}$ column), and (c) $\sum_{k} \pi_{ijk} \le 1\ \forall i,j$ (buyer $ij$ is offered only one price in every realization).

We define {\em soft-OPM} as a type of posted pricing mechanism, where (a) the prices for buyers are drawn independently, (b) it allocates items to all buyers who have valuations above their respective offered prices, and (c) satisfies feasibility constraints {\em in expectation} but not in every realization. In other words, a soft-OPM sets prices so that it sells to at most one buyer in expectation in each column and in each row.
It is obvious that the utility obtained by a soft-OPM is independent of the arrival order of buyers, and the revenue from a buyer is independent of past or future events in a realization of the mechanism.

We note an important property of soft-OPMs in the following lemma. It enables the redistribution of probability mass for any given price across buyers, and makes it sufficient to guess the total sale probability at each price. It is the analogue of Lemma~\ref{lem:mass-general} for multi-unit auctions. Its proof is practically identical to that of Lemma \ref{lem:mass-general} -- it applies a series of split operations followed by a series of merge operations to transform an optimal mechanism to a given soft-OPM. The only difference is that the submodular function in this case is $\ut{\infty}$ instead of $\ut{k}$, that is, there is no explicit bound on the inventory. So we skip the proof of this lemma.

\begin{lemma}\label{lem:mass-unit-demand}
Consider any soft-OPM $\M_1$, that offers price  $p_k\in \P$ with probability $\pi'_{ijk}$, and satisfies the following condition:
$\sum_{i,j}\pi'_{ijk}(1-F(p_k)) =\sum_{i,j}\pi_{ijk}(1-F(p_k))= q_k, \ \forall p_k\in\P$ (for every price, the {\em total sale probability}, {\em i.e.} the sale probabilities summed over all buyers, at a given price in $\M_1$ is the same as that in $\M_{\OPTs}$).
Then $\u(\M_1)\ge (1-\frac{1}{e})\OPT$.
\end{lemma}

\smallskip
\noindent
{\bf Algorithm: } We shall design a polynomial time algorithm to compute a soft-OPM $\M_1$ such that $\u(\M_1)\ge (1-\frac{1}{e}-O(\eps))\OPT$. We shall then construct a partial-OPM $\M_2$ from $\M_1$, as follows:  if buyer $ij$ is offered price $p_k$ w.p. $\pi'_{ijk}$ in $\M_1$, then $\M_2$ offers price $p_k$ to $ij$ w.p. $\pi'_{ijk}/3$ and a price of $\infty$ with the remaining probability (provided that  serving buyer $ij$ will not violate feasibility of matrix constraint). Lemma \ref{lem:profit-opm} below implies that $\M_2$ gives us the required guarantee, so it only remains to design an algorithm to compute a soft-OPM. To compare our result with the $1/6.75$-approximation for expected revenue \cite{CHMS10}, one may note that for linear utility, a reduction to soft-OPM such as Lemma \ref{lem:mass-unit-demand} holds quite easily without losing a factor of $(1-\frac{1}{e})$, while Lemma \ref{lem:profit-opm} naturally remains true (moreover, it even holds for complete OPMs).

\begin{lemma}\label{lem:profit-opm}
For any soft-OPM $\M_1$, if we construct a partial-OPM $\M_2$ from $\M_1$ as above, then $\u(\M_2)\ge \u(\M_1)/6.75$ for any ordering among the rows.
\end{lemma}
\begin{proof}
We first note that a soft-OPM $\M_3$ that offers price $p_k$ to buyer $ij$ w.p. $\pi'_{ijk}/3$ has expected utility $\u(\M_3)\geq \u(M_1)/3$ since $\u$ is concave. We shall now show that $\u(\M_2)\ge \frac{4}{9} \u(\M_3)$.

Given an order $\sigma$ on the rows of the matrix, consider an arbitrary buyer $ij$. Fix a realization of prices and valuations $({{\bf p,v}})$ of all buyers that belong to rows $\sigma(k)$ for $k<\sigma^{-1}(i)$ (buyers from previous rows than row $i$ according to $\sigma$), other than buyers belonging to column $j$, let this be event $E_{{\bf p,v}}$. We define event $E'_{{{\bf p,v}}}$ as the event where $E_{{{\bf p,v}}}$ occurs, and in addition, and in the set of buyers that arrive before $ij$, no buyer from the $i^{th}$ row or $j^{th}$ column has enough valuation to accept its offered price. If  $E'_{{{\bf p,v}}}$ occurs, then the seller will definitely be able to sell to buyer $ij$ if the latter can accept the offered price.

We note that in event $E'_{{{\bf p,v}}}$, the revenue generated by $\M_2$ from all buyers that come before $ij$ is a fixed value, let it be $\Rev_{{{\bf p,v}}}$. This is because the valuations and prices for all buyers that come before $ij$, other than those in row $i$ and column $j$, are fixed, and there is no sale in row $i$ and column $j$ until buyer $ij$ arrives. Let $X_{{\bf p,v}}$ be a random variable indicating the revenue generated by $\M_3$ in event $E_{{{\bf p,v}}}$ before buyer $ij$ arrives.  As $\M_3$ is a soft-OPM, every sale that happens in $\M_2$ in event $E'_{{{\bf p,v}}}$ also happens in $\M_3$ in event $E_{{{\bf p,v}}}$. So we get $X_{{\bf p,v}}\ge \Rev_{{{\bf p,v}}}$ in every realization.

The probability that no buyer in the $i^{th}$ row, except perhaps $ij$, has a valuation exceeding its offered price, is at least $2/3$, by Markov's inequality. The same holds for the $j^{th}$ column. Moreover, since these two events are independent, we get that
$\Pr{E'_{{{\bf p,v}}}}\ge \frac{4\Pr{E_{{{\bf p,v}}}}}{9}$.

Let $R_{ij}$ be a random variable that indicates the revenue from buyer $ij$ in $\M_3$. We note that $R_{ij}$ is independent of $E_{{{\bf p,v}}}$. The revenue from $ij$ in $\M_2$ has same distribution as $R_{ij}$ in event $E'_{{{\bf p,v}}}$, and it is $0$ otherwise.
The contribution to utility of buyer $ij$ in $\M_3$ is\\
$\sum_{{\bf p,v}} \left(\u(R_{ij}+X_{{\bf p,v}})-\u(X_{{\bf p,v}})\right)\Pr{E_{{{\bf p,v}}}}$.
The contribution to utility of buyer $ij$ in $\M_2$ is\\
$\sum_{{\bf p,v}} \left(\u(R_{ij}+\Rev_{{{\bf p,v}}})-\u(\Rev_{{{\bf p,v}}})\right)\Pr{E'_{{{\bf p,v}}}}$.
Since $\u$ is concave, so in every realization,
$$\left(\u(R_{ij}+\Rev_{{{\bf p,v}}})-\u(\Rev_{{{\bf p,v}}})\right) \ge \left(\u(R_{ij}+X_{{\bf p,v}})-\u(X_{{\bf p,v}})\right)$$

So buyer $ij$'s contribution to utility in $\M_2$  is at least $4/9$ times his contribution in $\M_3$. Summing over contributions of all buyers in every event, we get that $\u(\M_2)\ge \frac{4}{9}\u(\M_3)$.
\end{proof}

\smallskip

\noindent
{\bf Algorithm to compute soft-OPM: } It is similar to the algorithm designed for computing SPM for $k$-unit auction.

We divide the prices in $\M_{\OPTs}$ into $3$ classes, {\em small}, {\em large} and {\em huge}. Let $\P_{\hg}$ be the set of {\em huge prices} $p_k\ge \u^{-1}(\OPT/\eps)$. The distinction between small and large prices depend more intricately on the optimal mechanism. Let $p^*$ be the largest price such that $\sum_{\u^{-1}(\OPT/\eps) > p_k\ge p^*} q_k \ge 1/\eps^4$, {\em i.e.} the threshold where the total sale probability of all large prices add up to at least $1/\eps^4$. If such a threshold does not exist, then let $p^*=0$. Let $\P_{\sm}$ be all prices less than $p^*$, so that $\P_{\lrg}=\{p_k \mid \u^{-1}(\OPT/\eps) > p_k\ge p^*\}$. Now, using Lemma \ref{lem:new} in the proof of Lemma~\ref{lem:mass-unit-demand} allows us to prove the following variant of Lemma \ref{lem:mass-unit-demand}. The proof of Lemma \ref{lem:mass-unit-demand-variant} is practically identical to that of Lemma \ref{lem:mass-variant}, so we omit it.

\begin{lemma}\label{lem:mass-unit-demand-variant}
Consider any soft-OPM $\M_1$, that offers price $p_k$ to buyer $ij$ w.p. $\pi'_{ijk}$, such that (a) $\sum_{i,j,p_k\in \P_{\sm}} p_k \pi'_{ijk} (1-F_{ij}(p_k))  =\sum_{p_k\in \P_{\sm}} p_k q_k$, (b) for each $p_k\in\P_{\lrg}$, $\sum_{i,j} \pi'_{ijk} (1-F_{ij}(p_k))=q_k$, and (c) $\sum_{i,j,p_k\in \P_{\hg}} \u(p_k) \pi'_{ijk} (1-F_{ij}(p_k)) =\sum_{p_k\in \P_{\hg}} \u(p_k) q_k$. Then we have $\u(\M_1)\ge (1-\frac{1}{e}-O(\epsilon))\OPT$.
\end{lemma}

Thus it suffices to guess (a) the threshold $p^*$ separating large and small prices, (b) the total probability of sale $q_k$ for each large price $p_k$, (c) the expected revenue from small prices $\Rev$, and (d) the expected contribution to utility, say $h$, from huge prices.
Just as in the case of multi-unit auctions, the large prices (and hence the threshold $p^*$ as well) can be discretized so that they come only from the set $\u^{-1} (\eps^5 t \OPT)$ for any integer $0\le t\le \frac{1}{\eps^6}$.
Similarly, $h$ and $\u(\Rev)$ can be guessed up to a multiple of $\eps\OPT$, and the total sale probability for each large price to the nearest multiple of $\epsilon^{8}$. All these guesses define a {\em configuration}. The number of configurations is bounded by $2^{poly(1/\eps)}$.

{\small
\begin{eqnarray*}
{\textstyle \sum_{i,j}} \ \ \left(1-F_{ij}(p_k)\right)x_{ijk} &\ge& q_k \ \ \forall p_k\in \P_{\lrg}\\
{\textstyle \sum_{i,j,p_k\in \P_{\sm}}} \ \ \left(1-F_{ij}(p_k)\right) p_k x_{ijk} &\ge& \Rev  \\
{\textstyle \sum_{i,j,p_k\in \P_{\hg}}}\ \ (1-F_{ij}(p_{k}))\u(p_k) x_{ijk} &\ge& h \\
{\textstyle \sum_{i,k}}\ \ \left(1-F_{ij}(p_k)\right)x_{ijk} &\le& 1 \ \ \forall j\\
{\textstyle \sum_{j,k}}\ \ \left(1-F_{ij}(p_k)\right)x_{ijk} &\le& 1 \ \ \forall i\\
{\textstyle \sum_{k}}\ \ x_{ijk}  &\le& 1 \ \ \forall i,j\\
x_{ijk} &\ge& 0\ \  \forall i,j,k
\end{eqnarray*}
}

In the above LP, the variable $x_{ijk}$ denotes the probability that buyer $ij$ is shown price $p_k$. Each feasible configuration gives us a soft-OPM and we select the one with maximum utility.

\end{document}